\documentclass[a4paper,7pt]{article}
\usepackage{soul}
\usepackage[usenames,dvipsnames]{color}
\usepackage{jheppub}
\usepackage{esvect}
\usepackage{relsize}
\usepackage{comment}
\usepackage{amsmath, amssymb, slashed, epsf, color, graphicx, latexsym, bm}
\usepackage{mathrsfs}   
\usepackage{tensor}
\usepackage{slashed}
\usepackage{fixme}
\fxsetup{status=draft}
\usepackage{physics}
\usepackage{tikz}
\usepackage{tikz-cd}
\usetikzlibrary{positioning, calc}
\usepackage{dsfont}
\usepackage{pgfplots}
\pgfplotsset{compat=1.18}
\usepackage{epsfig}
\usepackage{graphics}
\usepackage[T1]{fontenc}
\usepackage{mathtools}
\usepackage{fixme}
\fxsetup{status=draft}
\usepackage{caption}
\usepackage{float}
\usepackage[bb=boondox]{mathalfa}
\usepackage{bbm}
\usepackage{subcaption}
\usepackage{enumitem}
\usepackage{verbatim}
\usepackage{blindtext}
\makeatletter
\renewcommand*\env@matrix[1][*\c@MaxMatrixCols c]{%
  \hskip -\arraycolsep
  \let\@ifnextchar\new@ifnextchar
  \array{#1}}
\makeatother
\DeclareFontEncoding{LS1}{}{}
\DeclareFontSubstitution{LS1}{stix}{m}{n}
\DeclareSymbolFont{symbols2}{LS1}{stixfrak}{m}{n}
\DeclareMathSymbol{\typecolon}{\mathbin}{symbols2}{"25}
\usepackage[mathlines]{lineno}

\newcommand{\be}{\begin{equation}}
\newcommand{\ee}{\end{equation}}
\usepackage{environ}

\NewEnviron{eqsp}{%
  \begin{equation}
    \begin{split}
      \BODY
    \end{split}
  \end{equation}
}

\def\CC {{\cal C}}

\def\CF {{\cal F}}

\def\CL {{\cal L}}

\def\CO {{\cal O}}

\def\CR {{\cal R}}

\def\CO {{\cal O}}

\def\CS {{\cal S}}

\def\IH{\mathbb{H}}

\def\IN{\mathbb{N}}

\def\IR{{\mathbb{R}}}

\def\IX{{\mathbb{X}}}
\def\IZ{{\mathbb{Z}}}

\def\F\IX{\mathfrak{x}}

\def\F\IX{\mathfrak{X}}

\def\r0{\lceil r_0\rceil}
\def\ben{\begin{eqnarray}}
\def\een{\end{eqnarray}}
\def\eps{\epsilon}

\def\vt{\vartheta}

\def\figfour{

\def\JPicScale{1}
\ifx\JPicScale\undefined\def\JPicScale{1}\fi
\unitlength \JPicScale mm
\begin{picture}(130,75)(0,0)
\linethickness{0.3mm}
\qbezier(30,70)(29.98,59.59)(31.19,52.38)
\qbezier(31.19,52.38)(32.39,45.16)(35,40)
\qbezier(35,40)(37.58,34.8)(41.19,30.59)
\qbezier(41.19,30.59)(44.8,26.38)(50,22.5)
\qbezier(50,22.5)(55.19,18.56)(60,18.56)
\qbezier(60,18.56)(64.81,18.56)(70,22.5)
\qbezier(70,22.5)(75.2,26.38)(78.81,30.59)
\qbezier(78.81,30.59)(82.42,34.8)(85,40)
\qbezier(85,40)(87.61,45.16)(88.81,52.38)
\qbezier(88.81,52.38)(90.02,59.59)(90,70)
\linethickness{0.3mm}
\put(60,18.7){\line(0,1){50}}
\linethickness{0.3mm}
\multiput(60,18.4)(0.15,0.13){167}{\line(1,0){0.15}}
\linethickness{0.3mm}
\multiput(35,40)(0.15,-0.130){167}{\line(1,0){0.15}}
\linethickness{0.3mm}
\put(35,40){\line(0,1){30}}
\linethickness{0.3mm}
\put(85,40){\line(0,1){30}}
\put(65,75){\makebox(0,0)[cc]{$y := \sigma_2/\tau_2$}}

\put(110,13){\makebox(0,0)[cc]{$x := z_2/\tau_2$}}

\put(50,45){\makebox(0,0)[cc]{$\CR$}}

\put(70,45){\makebox(0,0)[cc]{$\CL$}}

\put(60,15){\makebox(0,0)[cc]{$0$}}

\put(77,23){\makebox(0,0)[cc]{${1\over 2}$}}

\put(86,35){\makebox(0,0)[cc]{${1}$}}

\put(34,35){\makebox(0,0)[cc]{${-1}$}}

\linethickness{0.3mm}
\multiput(60,18.4)(0.15,0.092){110}{\line(1,0){0.15}}

\multiput(76.5,28.5)(0.15,0.2){110}{\line(1,0){0.15}}

\linethickness{0.3mm}
\put(0,18.5){\line(1,0){130}}
\end{picture}

}

\def\figfive{

\def\JPicScale{1}
\ifx\JPicScale\undefined\def\JPicScale{1}\fi
\unitlength \JPicScale mm
\begin{picture}(130,75)(0,0)
\linethickness{0.3mm}
\qbezier(30,70)(29.98,59.59)(31.19,52.38)
\qbezier(31.19,52.38)(32.39,45.16)(35,40)
\qbezier(35,40)(37.58,34.8)(41.19,30.59)
\qbezier(41.19,30.59)(44.8,26.38)(50,22.5)
\qbezier(50,22.5)(55.19,18.56)(60,18.56)
\qbezier(60,18.56)(64.81,18.56)(70,22.5)
\qbezier(70,22.5)(75.2,26.38)(78.81,30.59)
\qbezier(78.81,30.59)(82.42,34.8)(85,40)
\qbezier(85,40)(87.61,45.16)(88.81,52.38)
\qbezier(88.81,52.38)(90.02,59.59)(90,70)
\linethickness{0.3mm}
\put(60,18.7){\line(0,1){50}}

\put(69,22){\line(0,1){46.7}}

\put(77,28.8){\line(0,1){39.9}}

\multiput(60,18.7)(0.15,0.055){60}{\line(1,0){0.15}}

\multiput(69,22)(0.15,0.13){54}{\line(1,0){0.15}}

\multiput(77,28.8)(0.15,0.22){55}{\line(1,0){0.15}}

\linethickness{0.3mm}
\linethickness{0.3mm}
\linethickness{0.3mm}
\linethickness{0.3mm}
\put(85,40){\line(0,1){28.6}}
\put(65,75){\makebox(0,0)[cc]{$y' := \sigma_2'/\tau_2'$}}

\put(110,13){\makebox(0,0)[cc]{$x' := z_2'/\tau_2'$}}

\put(65,45){\makebox(0,0)[cc]{$\CL$}}

\put(60,15){\makebox(0,0)[cc]{$0$}}

\put(77,26){\makebox(0,0)[cc]{${2}$}}

\put(71,21){\makebox(0,0)[cc]{${1}$}}

\put(86,37){\makebox(0,0)[cc]{${3}$}}

\linethickness{0.3mm}

\linethickness{0.3mm}
\put(0,18.5){\line(1,0){130}}
\end{picture}

}

\newcommand{\refb}[1]{(\ref{#1})}

\newcommand{\mr}[1]{\mathrm{#1}}

\newcommand{\wt}[1]{\widetilde{#1}}

\begin{document}
\title{Generating Function of Single Centered Black Hole Index from the
Igusa Cusp Form}
\author{Ajit Bhand$^a$, Ashoke Sen$^b$, Ranveer Kumar Singh$^c$}
\affiliation[a]{Department of Mathematics, Indian Institute of Science Education and Research Bhopal Bhopal bypass, Bhopal 462066, India}
\affiliation[b]{International Centre for Theoretical Sciences - TIFR, Bengaluru - 560089, India}
\affiliation[c]{NHETC and Department of Physics and Astronomy, Rutgers University, 
126 Frelinghuysen Rd., Piscataway NJ 08855, USA}
\emailAdd{abhand@iiserb.ac.in, ashoke.sen@icts.res.in, ranveersfl@gmail.com}

\abstract{We introduce manifestly duality invariant
generating function of the index of single centered black
holes in the heterotic string theory compactified on a six dimensional torus.
This function is obtained by subtracting, from the inverse of the Igusa
cusp form, the generating function of the index of two centered black
holes constructed from the Dedekind eta function. 
We also study the analytic
properties of this function in the Siegel upper half plane.} 
\def\rem{$\clubsuit$}
\newcommand{\cas}[1]{\marginpar{\raggedright \tiny \rem AS: #1 \rem}} 
\newcommand{\caj}[1]{\marginpar{\raggedright \tiny \textcolor{blue}{\rem AB: #1 \rem}}}    
\newcommand{\cran}[1]{\marginpar{\raggedright \tiny \textcolor{OliveGreen}{\rem RKS: #1 \rem}}}    

\maketitle

\section{Introduction and summary}
In the
heterotic string theory on $T^6\times \IR^{3,1}$, the BPS index of a class of quarter BPS dyons is known to be given by the Fourier
coefficients of the inverse of the Igusa cusp form $\Phi_{10}$ -- a weight ten
Siegel modular form for 
$\mr{Sp}(2,\mathbb{Z})$ \cite{Dijkgraaf:1996it,Shih:2005uc,LopesCardoso:2004law,David:2006yn}. 
These Fourier coefficients can be expressed as appropriate Fourier integrals
of $1/\Phi_{10}$.
The correspondence is as follows. 
The black holes in heterotic string theory on $T^6\times \IR^{3,1}$ are labeled by a
28 dimensional electric charge vector  $Q$ and a 28 dimensional magnetic charge vector
$P$. However, for torsion one dyons, for which
\be
\gcd\{ Q_i P_j - Q_j P_i\} = 1\, ,
\ee
the index is a function of only three quadratic 
combination of charges which we shall denote by $Q^2,P^2$ and $Q\cdot P$ \cite{Dabholkar:2007vk,Banerjee:2007sr}. We define
\be
m:= Q^2 / 2, \qquad n := P^2 / 2, \qquad \ell :=Q\cdot P\, , \qquad m,n,\ell\in \mathbb{Z}\,.
\ee
Let us denote the index by $d(m,n,\ell)$. Then we have
\be \label{edmnl}
d(m,n,\ell) = (-1)^{\ell+1} \int_{\mathcal{C}}
d\tau d\sigma dz\, e^{-2 \pi i\left(m\tau+n\sigma+\ell z\right)}
\frac{1}{\Phi_{10}(\tau,\sigma,z)}\, ,
\ee
where 
$\tau,\sigma,z$ are  three complex
variables 
that
parametrize the Siegel upper half space $\IH_2$:
\be
\tau_2>0, \qquad \sigma_2>0, \qquad \tau_2\sigma_2>z_2^2\, ,
\ee
where we have defined
\be
(\tau_2,\sigma_2,z_2)= {\rm Im}(\tau,\sigma,z), \qquad (\tau_1,\sigma_1,z_1)
= {\rm Re}(\tau,\sigma,z)\, .
\ee
The integration is performed over $\tau_1,\sigma_1,z_1$ from 0 to 1
keeping fixed $\tau_2,\sigma_2,z_2$. $\CC$ carries information about the values of $\tau_2,\sigma_2,z_2$ along the integration contour. 
We shall introduce the notations
\be
\Omega := \begin{pmatrix} \tau & z\cr z & \sigma \end{pmatrix}~,
\ee
and
\be
T  := \begin{pmatrix}
    m & \ell/2\\ \ell/2 & n
\end{pmatrix}\, ,
\ee
and use $\Omega$ to label the variables $(\tau,\sigma,z)$ and $T$ to label the
charges $(m,n,\ell)$.

There are however some subtleties in this result.
Heterotic string theory on $T^6\times \mathbb{R}^{3,1}$ has a moduli space. 
For a given set of
charges $(m,n,\ell)$,
the moduli space
can be divided into chambers separated by codimension one walls known as
walls of marginal stability.
While the index remains constant inside a given chamber,
it can undergo jumps as we move from one chamber to another \cite{Sen:2007vb,Dabholkar:2007vk}. 
These jumps
can be traced to the fact that for a given set of charges $(m,n,\ell)$, the index gets
contribution from both, single centered black holes and
the two centered black holes carrying the same total charge. While the contribution to the
index from single centered black holes is independent of the moduli, the two centered
black holes can cease to exist as we cross a wall of marginal stability
from one chamber to another, and for this reason
their contribution to the index can jump \cite{Sen:2007pg,Cheng:2007ch}. 
On the other hand since $1/\Phi_{10}$ 
has poles, its Fourier transform is also not unique since the Fourier integrals of
$1/\Phi_{10}$ depends on the choice of integration contour. In
particular, the $(\tau_2,\sigma_2,z_2)$ space labelling the integration contour 
can also be divided into chambers and the
Fourier integrals performed in different chambers give different results. 
It turns out \cite{Sen:2007vb} that for a given set of charges,
there is a one-to-one map from the
chambers in the moduli space to the chambers in the $(\tau_2,\sigma_2, z_2)$ space
in the region where
\be
\det \rm Im\, \Omega>{1\over 4} \, .
\ee
In this region the walls separating the different chambers  in the $(\tau_2,\sigma_2, z_2)$ space
lie along
\be \label{ewall}
m_1\tau_2 - n_1\sigma_2 + j\, z_2 = 0, \qquad m_1,n_1,j\in\mathbb{Z}, \qquad
m_1n_1+{j^2\over 4}
={1\over 4}\, .
\ee
The Fourier coefficient $d(m,n,\ell)$ of $\Phi_{10}^{-1}$
in a given chamber in the $(\tau_2,\sigma_2,z_2)$ space, defined by
\refb{edmnl} with $\tau_2,\sigma_2,z_2$ lying inside that chamber,  
gives the index in the
corresponding chamber in the moduli space. 

As discussed before, the index typically includes contribution
from both single centered black holes and two centered black holes. However,
for a given $(m,n,\ell)$ satisfying
\be \label{erangepositive}
m\ge 0, \qquad n\ge 0, \qquad 4mn-\ell^2\ge 0\, ,
\ee
there is a special chamber in the $(\tau_2,\sigma_2,z_2)$ space, known as the
attractor chamber, 
where the contribution from two centered black holes vanish and we get single centered
black hole index $d^*(m,n,\ell)$ carrying charges $(m,n,\ell)$. 
The precise expression for $d^*(m,n,\ell)$ is \cite{Cheng:2007ch}
\begin{equation}\label{edstarTint}
\wt d^{*}(m,n,\ell) =\left\{ \begin{split}
&(-1)^{\ell+1} \int_{\mathcal{C}_{m,n,\ell}} d\tau d\sigma dz\, e^{-2 \pi i\left(m\tau+n\sigma+\ell z\right)} \frac{1}{\Phi_{10}(\Omega)}~, \ \ \hbox{{for} $m\ge 0$, $n\ge 0$, $4mn-\ell^2\ge 0$~,} \\
& 0~, \qquad \hbox{\rm otherwise}~, \,
    \end{split}\right. 
\end{equation}
where the contour $\mathcal{C}_{m,n,\ell}$ is given as 
\begin{equation}
\begin{split}
\mathcal{C}_{m,n,\ell}:\quad 
\tau_2=\frac{2n}{\varepsilon}, \quad \sigma_2=\frac{2m}{\varepsilon}, \quad z_2=-\frac{\ell}{\varepsilon},\qquad
0 \leq \tau_1,\sigma_1,z_1<1~.
\end{split}
\label{eq 3.3}
\end{equation}
Here $\varepsilon>0$ is a real positive number, sufficiently small so that  
$\det \rm Im \, \Omega$ is 
larger than $1/4$.

Special attention is needed for zero discriminant states satisfying $4mn-\ell^2=0$, $m,n\ge 0$,
since in this
case \refb{eq 3.3} gives $\tau_2\sigma_2-z_2^2=0$ where $\Phi_{10}$ is ill defined. We resolve
this issue by deforming $m,n,\ell$ in \refb{eq 3.3} by small real numbers so that $(\tau_2,\sigma_2,z_2)$
is brought into the interior of $\IH_2$ where $\Phi_{10}$ is well defined and,
at the same time, take $\varepsilon$ to be small enough so that $\det \rm Im \, \Omega$
is larger than $1/4$: 
\be \label{edeformedcontour}
\tau_2=\frac{2n+\varepsilon_1}{\varepsilon}, \quad \sigma_2=\frac{2m+\varepsilon_2}{\varepsilon}, \quad 
z_2=-\frac{\ell+\varepsilon_3}{\varepsilon}, \qquad 0<|\varepsilon_1|,|\varepsilon_2|,|\varepsilon_3|\ll 1\, .
\ee
We
then define $d^*(m,n,\ell)$ using \refb{edstarTint}. One may worry that this may not give a unique
answer since $1/\Phi_{10}$ has poles on the `walls of marginal stability', 
and different deformations of $m,n,\ell$ may land us on
different sides of a wall of marginal stability, producing different results for the integral in
\refb{edstarTint}. We shall show in Section \ref{sgenerating} that this does not happen and we get an unambiguous result
for $d^*(m,n,\ell)$. The same procedure needs to be applied to deal with the poles of $1/\Phi_{10}$
that lie on the contour $\mathcal{C}_{m,n,\ell}$ for general $m,n,\ell$, {\it e.g.} for $\ell=0$ the
pole at $z=0$ lies on the integration contour. We simply deform the 
contour to avoid the pole and the result does not depend on how we deform the contour.

Equation \refb{eq 3.3} can be written as
\begin{equation}\label{edstarTintOmega}
d^{*}(T) =\left\{\begin{split}
    &(-1)^{T_{12}+1} \int_{{\rm Im}\, \Omega = 2\gamma_0^t T \gamma_0/\varepsilon} 
d^3( {\rm Re}\,\Omega)\, e^{-2 \pi i{\rm Tr} (\Omega T)} \frac{1}{\Phi_{10}(\Omega)}~,\quad 
\hbox{for $T\ge 0$}~, \\
&0~, \qquad \hbox{otherwise}\, ,
    \end{split}\right. 
\end{equation}
where $T\ge 0$ means that $T$ has non-negative eigenvalues and,
\be
\gamma_0:= \begin{pmatrix} 0 & -1 \cr 1 & 0\end{pmatrix}\, .
\ee
It follows from \refb{edstarTintOmega} and the identities
\be
\Phi_{10} (\gamma\Omega \gamma^t) = \Phi_{10}(\Omega), \quad
\gamma \gamma_0 \gamma^t=\gamma_0, \quad \gamma^t \gamma_0 \gamma=\gamma_0~,\quad d^3(\mr{Re}\,\Omega)=d^3(\mr{Re}\,\gamma\Omega\gamma^t)~, 
\quad \hbox{for} \ \gamma\in \mathrm{PSL}(2,\IZ)\, ,
\ee
that $d^*$ is invariant under an $\mr{SL}(2,\mathbb{Z})$ transformation: 
\be\label{esl2zT}
d^*(\gamma^t T\gamma) = d^*(T), \qquad \gamma\in \mr{SL}(2,\mathbb{Z})\, .
\ee
For $\gamma=\begin{pmatrix} a & b\cr c & d\end{pmatrix}$ the transformation
$T\mapsto \gamma^t T \gamma$ takes the form:
\be\label{e2}
\begin{pmatrix}
    m\\ n\\ \ell
\end{pmatrix}\mapsto \begin{pmatrix}
    a^2 m + c^2 n + ac \ell\cr b^2 m + d^2 n + bd\ell\cr
2 ab m + 2 cd n + (ad + bc)\ell
\end{pmatrix}\, .
\ee

Our focus in this paper will be on the generating function of single centered
black holes. We define
\be\label{edefFO}
F(\Omega) = \sum_{m,n,\ell\in {\mathbb{Z}}} d^*(m,n,\ell)\, (-1)^{\ell+1}\, 
e^{2 \pi i\left(m\tau+n\sigma+\ell z\right)}=\sum_T (-1)^{2 T_{12}+1}\, d^*(T) e^{2\pi i {\rm Tr}(\Omega T)}\, .
\ee
If we can prove the absolute convergence of $F$, then it follows from \refb{edefFO} and \refb{esl2zT} that 
$F(\Omega)$ is also $\mr{SL}(2,\mathbb{Z})$ invariant:
\be
F(\gamma\Omega\gamma^t)=F(\Omega)\, .
\ee

We also consider a second way of defining the single centered index:
\begin{eqsp}\label{eq:deg_multi_single}
\wt d^*(m,n,\ell)=d(m,n,\ell)-d^{\text{two}}(m,n,\ell)~,    
\end{eqsp}
where $d(m,n,\ell)$ is the total index in a given chamber in the moduli space,
computed from \refb{edmnl} and 
$d^{\text{two}}(m,n,\ell)$ denotes the contribution to the index from
two-centered black holes in the same
chamber of the moduli space.
We introduce the generating function for $\wt d^*$:
\be\label{edeftildeF}
\wt F(\Omega) = \sum_T \, (-1)^{2 T_{12}+1} \, \wt d^*(T) \, e^{2\pi i \mr{Tr}(\Omega T)}\, .
\ee
From the results on $d^{\text{two}}(m,n,\ell)$ found in 
\cite{Dabholkar:2007vk,Sen:2007pg,Cheng:2007ch,1104.1498,1210.4385,
Chowdhury:2019mnb,LopesCardoso:2020pmp} one gets
the following expression for $\wt F(\Omega)$:
\begin{eqnarray}
\label{eguessfinintro}
\wt F(\Omega) &=& {1\over \Phi_{10}(\Omega)}
- {1\over 2} \sum_{\big{(}\begin{smallmatrix} a & b\cr c & d\end{smallmatrix}\big{)}\in 
\mr{PSL}(2,\mathbb{Z})}
\left(e^{\pi i \{ac\tau + bd\sigma + (ad+bc)z\}} - e^{-\pi i \{ac\tau + bd\sigma + (ad+bc)z\}}\right)^{-2} \nonumber \\ && 
\hskip 1in 
\times\ f_+(a^2\tau +b^2\sigma +2abz) \ f_+(c^2\tau+d^2\sigma+2cd z) \nonumber \\ 
&-&  \sum_{p\ge 0} f_p f_{-1} \sum_{r>0} r 
\sum_{\big{(}\begin{smallmatrix} a & b\cr c & d\end{smallmatrix}\big{)}\in  \mr{PSL}(2,\mathbb{Z})}
 H(ac\tau_2 + bd\sigma_2 + (ad+bc)z_2) \nonumber \\ && \hskip 1in \times\
H\left(-ac\tau_2 - bd\sigma_2 - (ad+bc)z_2 + ra^2\tau_2 + rb^2\sigma_2+2rab z_2
\right) \,
\nonumber \\ && \hskip 1in \times 
e^{2\pi i \{(pa^2-c^2+r ac)\tau+(pb^2-d^2+r bd)\sigma+
(2pab-2cd + r(ad+bc) )z\}} 
\nonumber \\ 
&-&  f_{-1}^2 \sum_{r>0} r 
\sum_{\big{(}\begin{smallmatrix} a & b\cr c & d\end{smallmatrix}\big{)}\in G_r\backslash \mr{PSL}(2,\mathbb{Z})}\hskip .1in 
 \bigg\{\prod_{n=-\infty}^\infty 
H(a_nc_n\tau_2 + b_nd_n\sigma_2 + (a_nd_n+b_nc_n)z_2)
\bigg\}
\nonumber \\ && \hskip 1in \times \
e^{2\pi i \{(-a^2-c^2+r ac)\tau+(-b^2-d^2+r bd)\sigma+
(-2ab-2cd + r(ad+bc)) z\}}  \, ,
\end{eqnarray}
where $H$ is the Heaviside function
\be
H(x) = \begin{cases} \hbox{1 for $x > 0$}\cr
\hbox{0 for $x\le 0$} \end{cases} \, ,
\ee
the coefficients $f_p$ and the function $f_+(\tau)$ are defined through the expansion
\be \label{e122}
\eta(\tau)^{-24} = \sum_{p=-1}^\infty f_p \, e^{2\pi ip \tau}, \qquad f_+(\tau) 
= \sum_{p=0}^\infty f_p \, e^{2\pi i p\tau}\, ,
\ee
where $\eta$ is the Dedekind eta function, 
$G_r$ is the cyclic group generated by 
\be
\begin{pmatrix}0 & -1\cr 1 & -r \end{pmatrix}~,
\ee
and 
\be\label{edefanbnintro}
\begin{pmatrix} a_n & b_n\cr c_n & d_n\end{pmatrix}:=\begin{pmatrix} 0 & -1\cr 1 & -r\end{pmatrix}^n
\begin{pmatrix} a & b\cr c & d\end{pmatrix}\, .
\ee
In Section \ref{section3}, we shall provide a physical derivation of \eqref{eguessfinintro} based on \cite{Dabholkar:2007vk,Sen:2007pg,Cheng:2007ch,1104.1498,1210.4385,Chowdhury:2019mnb}.

We prove the following results:
\begin{enumerate}
\item In Section \ref{sgenerating}, we show that $d^*(m,n,\ell)$ defined in \refb{edstarTint} is unambiguous even when $T$ has a zero
eigenvalue and when a pole falls on the contour $\mathcal{C}_{m,n,\ell}$ since the residue
at the pole vanishes.
\item In Section \ref{sgenerating}, we also argue that the sum over $m,n,\ell$ in \refb{edefFO} converges absolutely and uniformly on compact subsets of the domain $\det \rm Im \, \Omega>1/4$.
As a result $F(\Omega)$ is analytic for $\det \rm Im \, \Omega>1/4$.
\item In Section \ref{stildeFconverge}, in Theorem \ref{thm:S_conv}, we show that the sum over $a,b,c,d,r$ in each of the terms in
\refb{eguessfinintro} converges absolutely and uniformly on compact subsets of
$\IH_2$ except on the subspaces
\be\label{epoleint}
m_1\tau - n_1\sigma + m_2 + j\, z = 0, \qquad
m_1,n_1,m_2\in \mathbb{Z}, \qquad m_1n_1  +
{j^2\over 4} = {1\over 4}\, .
\ee
\item In Theorem \ref{thm:S_mero_poles}, we show that on the subspaces \refb{epoleint} 
the sums in \refb{eguessfinintro}  have poles.
In theorem \ref{thm:F_rel_Igusa} we show that these poles precisely cancel the poles
of $1/\Phi_{10}(\Omega)$ at \refb{epoleint}.
\item In Theorem \ref{thm:sing_deg_gen}, we show that despite the presence of the Heaviside functions in its definition,
$\wt F(\Omega)$ admits a meromorphic continuation to $\IH_2$
with double poles at 
\begin{eqnarray}
&& n_2 (\tau\sigma - z^2) + m_1\tau - n_1\sigma + m_2 + j\, z = 0~, \nonumber \\
&&
m_1,n_1,m_2,n_2\in \mathbb{Z}, \qquad n_2\ge 1, \qquad m_1n_1 + m_2 n_2 +
{j^2\over 4} = {1\over 4}\, .
\end{eqnarray}
The behaviour at these poles coincide with those of $1/\Phi_{10}(\Omega)$. 
However $1/\Phi_{10}(\Omega)$ has additional poles for 
$n_2=0$ which are absent in $\wt F(\Omega)$.
\item In Theorem \ref{thm:sing_cent_att_cont}, we show that the two different definitions of single centered index and their generating functions
coincide:
\be\label{etildeFFequality}
d^*(T)=\wt d^*(T)\quad \ \forall \ T, \qquad F(\Omega)=\wt F(\Omega)\, ,
\ee
including contributions from states where $T$ has a negative eigenvalue.
This agrees with the analytical arguments of \cite{1104.1498} and numerical results of
\cite{Chowdhury:2019mnb}.
Since both $F(\Omega)$ and $\wt F(\Omega)$ are defined by analytic continuation from their
domains of convergence, a more precise statement will be that the analytic continuations
of $F(\Omega)$ and $\wt F(\Omega)$ agree.
\end{enumerate}

It is known that for fixed $m$, the coefficients $d^*(m,n,\ell)$ can be identified as Fourier
coefficients of a mock Jacobi form of index $m$ \cite{Dabholkar:2012nd}, {\it albeit} 
with some restrictions on $m,n,\ell$ \cite{Bhand:2023rhm,Banerjee:2025bqi}. In contrast $d^*(m,n,\ell)$ are the
Fourier coefficients of $F(\Omega)$ with no exceptions. $F(\Omega)$ has $\mathrm{PSL}(2,\mathbb{Z})$-invariance but does not have the full $\mathrm{Sp}(2,\mathbb{Z})$ invariance. It is tempting to
speculate that $F(\Omega)$ is closely related to a mock Siegel modular form \cite{bringmann2012kohnen,raum,bringmann2016almost}, although it is not
known at present if we can add appropriate non-holomorphic part to $F(\Omega)$ to
construct a harmonic form for $\mathrm{Sp}(2,\mathbb{Z})$.

While the analysis of this paper is restricted exclusively to the index of quarter BPS
black holes in heterotic string theory compactified in $T^6$, results similar to the ones
described in \refb{edmnl}, \refb{edstarTint} are also known to hold for more general 
class of string 
compactifications with sixteen unbroken supersymmetries
\cite{Jatkar:2005bh,David:2006ji,David:2006yn,David:2006ru,David:2006ud,Sen:2007qy}. 
We expect that even in these
general class of theories we should be able to define the generating function for single
centered black hole index and prove their convergence. We leave this for future work.

\section{Construction of the generating function  $\wt F(\Omega)$} \label{section3}

Our goal in this section is
to construct the generating function of the index of single-centered black holes by taking the difference between the complete generating function for black holes 
and the generating function for two-centered black holes. We
begin with $1/\Phi_{10}$ in a given chamber, and then subtract from it, the 
generating function of two centered black holes in the same chamber.

The inverse $\Phi_{10}^{-1}$ of the Igusa cusp form is given by \cite{Dijkgraaf:1996it,David:2006ji} 
\be \label{es5}
{1\over \Phi_{10}(\tau, \sigma, z)} 
=e^{-2\pi i(\tau+\sigma+z)}
\prod_{j,k,l\in \mathbb{Z}, \, 
4kl-j^2\ge -1\atop
k,l\ge 0, j<0 \, \hbox{\tiny for}\, k=l=0 }
\left( 1 - e^{2\pi i (k\tau+l\sigma + j z)}\right)^{-c(4kl - j^2)}~,
\ee
where the coefficients $c(s)$ are defined via the equation
\be \label{es6}
8\, \sum_{i=2}^4 {\vt_i(\tau, z)^2 \over \vt_i(\tau, 0)^2}
= \sum_{n,j} c(4n - j^2) \, e^{2\pi i (n\tau + j z)} \, ,
\ee
$\vt_i$'s being the Jacobi theta functions defined by the expnasion 
\begin{eqsp}\label{eq:theta_exp}
\begin{aligned}
& \vartheta_1(\tau, z)=\sum_{n \in \mathbb{Z}}(-1)^n q^{\frac{1}{2}\left(n-\frac{1}{2}\right)^2} \zeta^{n-\frac{1}{2}}, \quad \vartheta_2(\tau, z)=\sum_{n \in \mathbb{Z}} q^{\frac{1}{2}\left(n-\frac{1}{2}\right)^2} \zeta^{n-\frac{1}{2}} \\
& \vartheta_3(\tau, z)=\sum_{n \in \mathbb{Z}} q^{\frac{n^2}{2}} \zeta^n, \quad \vartheta_4(\tau, z)=\sum_{n \in \mathbb{Z}}(-1)^n q^{\frac{n^2}{2}} \zeta^n~,\quad q=e^{2\pi i\tau}~,~~\zeta=e^{2\pi i z}~.
\end{aligned}
\end{eqsp}
It follows from \refb{es6} and \eqref{eq:theta_exp} that
\be \label{es7}
c(s)=0  \quad \hbox{for} \quad s< -1\, , \qquad c(-1)=2\, .
\ee 
$\Phi_{10}$ is a Siegel modular form of weight ten of Sp(2,$\IZ$)
-- a holomorphic function on $\IH_2$ satisfying,
\begin{eqsp}
    \Phi_{10}((A\Omega+B)(C\Omega+D)^{-1})=[\mr{det}(C\Omega+D)]^{10}\Phi_{10}(\Omega)~,\quad \begin{pmatrix}
        A&B\\C&D
    \end{pmatrix}\in\mr{Sp}(2,\IZ)~.
\end{eqsp}

\noindent
$1/\Phi_{10}(\Omega)$ has double poles at \cite{Dijkgraaf:1996it,LopesCardoso:2004law,Sen:2007qy} 
\be \label{epoles}
n_2(\tau\sigma-z^2) - m_1\tau +n_1\sigma + jz
+ m_2=0, \quad
m_1,n_1,m_2,n_2,j\in \IZ,
\quad m_1 n_1+m_2n_2 + {j^2\over 4}={1\over 4}\, .
\ee
For $n_2\neq 0$, we can rewrite the hypersurface \eqref{epoles} as 
\begin{eqsp}\label{eq:poles_new}
    n_2\left[\left(\tau+\frac{n_1}{n_2}\right)\left(\sigma-\frac{m_1}{n_2}\right)-\left(z-\frac{j}{2n_2}\right)^2\right]+\frac{1}{4n_2}=0~.
\end{eqsp}
Let us write 
\begin{eqsp}
    \tilde{\tau}_1=\tau_1+\frac{n_1}{n_2}~,\quad \tilde{\sigma}_1=\sigma_1-\frac{m_1}{n_2}~,\quad \tilde{z}_1=z_1-\frac{j}{2n_2}~.
\end{eqsp}
Then the imaginary part of \eqref{eq:poles_new} gives 
\begin{eqsp}\label{eq:im_poles_new}
    \tilde{\tau}_1=\frac{2\tilde{z}_1z_2-\tau_2\tilde{\sigma}_1}{\sigma_2}~,
\end{eqsp}
and the real part of \eqref{eq:poles_new} gives 
\begin{eqsp}\label{eq:re_poles_new}
    n_2\left[\tilde{\tau}_1\tilde{\sigma}_1-\tau_2\sigma_2-\tilde{z}_1^2+z_2^2\right]+\frac{1}{4n_2}=0~.
\end{eqsp}
Combining \eqref{eq:im_poles_new} and \eqref{eq:re_poles_new} we get 
\begin{eqsp}\label{eq:im_re_poles_new}
{\tau_2\over \sigma_2}\tilde{\sigma}_1^2 + \tilde{z}_1^2 - 2\, {z_2\over \sigma_2}\,
\tilde{\sigma}_1 \tilde{z}_1 = {1\over 4n_2^2} - (\tau_2\sigma_2 - z_2^2)~.    
\end{eqsp}
The left hand side of \eqref{eq:im_re_poles_new}
is a positive definite quadratic form in the variables $\tilde{\sigma}_1,\tilde{\tau}_1$ for $\tau_2\sigma_2>z_2^2$. Therefore,
a solution to these equations exists only when the right hand side is positive, i.e.
when $\det \mathrm{Im} \,\Omega= \tau_2\sigma_2-z_2^2 \le 1/4n_2^2$. In particular, a sufficient condition for absence of poles with $n_2\neq 0$ is $\det \mathrm{Im} \,\Omega= \tau_2\sigma_2-z_2^2 > 1/4$.

This shows that, for  $\det\rm Im\,\Omega > 1/4$ we
only need to examine the poles corresponding to $n_2=0$ in \refb{epoles}. 
One of these, corresponding
to $m_1=m_2=n_1=0$, $j=1$, is the pole at $z=0$.
Near the double pole $z\to 0$, we have
\be\label{epolestructure}
{1\over \Phi_{10}(\Omega)} = (e^{\pi i z}- e^{-\pi i z})^{-2} f(\sigma) f(\tau)\, ,
\ee
where 
\be\label{edeff}
f(\tau)= \eta(\tau)^{-24} = \sum_{p\ge -1} f_p \, e^{2\pi i p\tau}\, ,
\ee
is the partition function of half-BPS states. The behaviour of $1/\Phi_{10}$ near other
double poles of the type given in \refb{epoles} with $n_2=0$ is obtained by $\mr{SL}(2,\mathbb{Z})$
transformation of \refb{epolestructure} given by $\Omega\to\gamma\Omega\gamma^t$ since one can
show that every other pole for $n_2=0$ can be related to the pole at $z=0$ using
the $\mathrm{PSL}(2,\mathbb{Z})$-map (see Lemma \ref{lemma:sl2z_linear_poles}). 
The Fourier coefficients of the RHS of \eqref{epolestructure} in a given chamber are known to give the contribution of two centered
black holes to the index and are responsible for the jump in the index
across the walls of marginal stability \cite{Sen:2007pg,Cheng:2007ch}. 
Thus the naive expectation would be that to get the generating function of
single centered index, we should remove from $1/\Phi_{10}$ the contribution on the right
hand side of \refb{epolestructure} and its $\mr{PSL}(2,\mathbb{Z})$ images.
This leads to the following guess for the generating function:
\begin{eqnarray}
\label{eguess}
&& {1\over \Phi_{10}(\Omega)}
- {1\over 2} \sum_{\big{(}\begin{smallmatrix} a & b\cr c & d\end{smallmatrix}\big{)}\in \mr{PSL}(2,\mathbb{Z})}
\left(e^{\pi i \{ac\tau + bd\sigma + (ad+bc)z\}} - e^{-\pi i \{ac\tau + bd\sigma + (ad+bc)z\}}\right)^{-2} \nonumber \\ && 
\hskip 1in 
\times\ f(a^2\tau +b^2\sigma +2abz) \ f(c^2\tau+d^2\sigma+2cd z) \nonumber \\ 
&=&  {1\over \Phi_{10}(\Omega)}
- {1\over 2} \sum_{\big{(}\begin{smallmatrix} a & b\cr c & d\end{smallmatrix}\big{)}\in \mr{PSL}(2,\mathbb{Z})}
\left(e^{\pi i \{ac\tau + bd\sigma + (ad+bc)z\}} - e^{-\pi i \{ac\tau + bd\sigma + (ad+bc)z\}}\right)^{-2}  \nonumber \\ &&
\hskip 1in \times \ \sum_{p,q=-1}^\infty f_p f_q \,
 e^{2\pi i \, p\{a^2\tau +b^2\sigma +2abz\}} \,  
e^{2\pi i q\{c^2\tau+d^2\sigma+2cd z\}} \, ,
\end{eqnarray}
where the factor of $1/2$ compensates for the fact that
the summand in (\ref{eguess}) remains 
invariant under $(a,b,c,d)\to (-c,-d,a,b)$ and hence each term effectively appears
twice. 

It is easy to see however that this cannot be the correct result.
For example, if we take $(a,b,c,d)=(1,0,c,1)$ with $c\in\mathbb{Z}$,
then for $q=-1$, the term in the third and fourth line of \refb{eguess} will involve a sum of the form:
\be \label{e314}
\sum_{c\in\mathbb{Z}} (e^{\pi i(c\tau+z)}-e^{-\pi i (c\tau+z)})^{-2} 
e^{-2\pi i (c^2\tau+\sigma + 2cz)}\,e^{2\pi ip\tau}\, .
\ee
Since $\tau$ has positive imaginary part, the sum over $c$ will diverge. 
As we shall see, the correct result will not suffer from such divergences.

We shall now carry out a more detailed analysis on why \refb{eguess} might 
differ from the partition function of two centered black hole states.
First consider the coefficient
of the terms involving $f_p f_q$ for $p,q\ge 0$. By expanding 
$\left(e^{\pi i \{ac\tau + bd\sigma + (ad+bc)z\}} - e^{-\pi i \{ac\tau + bd\sigma + (ad+bc)z\}}\right)^{-2}$
in a power series expansion in $e^{\pm\pi i \{ac\tau + bd\sigma + (ad+bc)z\}}$ in the respective domains
of convergence, we can express the term subtracted from $1/\Phi_{10}$ in \eqref{eguess} as
\begin{eqnarray}\label{e216xy}
&& -{1\over 2} 
\sum_{\big{(}\begin{smallmatrix} a & b\cr c & d\end{smallmatrix}\big{)}\in \mr{PSL}(2,\mathbb{Z})} 
 \sum_{r>0} r \bigg\{
e^{2\pi i r \{ac\tau + bd\sigma + (ad+bc)z\}} H(ac\tau_2 + bd\sigma_2 + (ad+bc)z_2) 
\nonumber \\ && \hskip 1in +\
e^{-2\pi i r (ac\tau + bd\sigma + (ad+bc)z)} H(-(ac\tau_2 + bd\sigma_2 + (ad+bc)z_2))\bigg\} 
\nonumber \\ && \hskip 1in \times\sum_{p,q\ge 0}f_p f_q e^{2\pi i p(a^2\tau+b^2\sigma+2abz)} e^{2\pi i q(c^2\tau+d^2\sigma+2cdz)}
~. 
\end{eqnarray} 

This can be reorganized as
\begin{eqnarray} \label{ereorg}
&& -{1\over 2}
\sum_{\big{(}\begin{smallmatrix} a & b\cr c & d\end{smallmatrix}\big{)}\in \mr{PSL}(2,\mathbb{Z})} 
\sum_{r>0}\sum_{p,q\ge 0} r  H(ac\tau_2 + bd\sigma_2 + (ad+bc)z_2) \,
f_p f_q \nonumber \\ && \hskip 1in \times 
e^{2\pi i \{(pa^2+qc^2+r ac)\tau+(pb^2+qd^2+r bd)\sigma+
(2pab+2qcd + r(ad+bc) )z\}}  \nonumber \\ &&
-{1\over 2}
\sum_{\big{(}\begin{smallmatrix} a & b\cr c & d\end{smallmatrix}\big{)}\in \mr{PSL}(2,\mathbb{Z})}
\sum_{r>0} \sum_{p,q\ge 0}r  H(-ac\tau_2 - bd\sigma_2 - (ad+bc)z_2) \,
f_p f_q \nonumber \\ && \hskip 1in \times 
e^{2\pi i \{(pa^2+qc^2-r ac)\tau+(pb^2+qd^2-r bd)\sigma+
(2pab+2qcd - r(ad+bc) )z\}} \, .
\end{eqnarray}
The term in the first two lines of \refb{ereorg} may be identified as the subtraction of the
contribution from a bound state of the
half-BPS states carrying charges $(aM, bM)$ and $(cN,dN)$, with
\be
M^2=2p, \quad N^2=2q, \quad M\cdot N=r, \quad r>0\, ,
\ee
so that the total charge $(Q,P)=(aM+cN,bM+dN)$
satisfies
\be
Q^2 = 2 (pa^2+qc^2+r ac), \quad P^2 = 2(pb^2+qd^2+r bd), \quad Q\cdot P =2pab+2qcd + r(ad+bc)\, .
\ee
Such bound states are known to exist in the chamber $ac\tau_2 + bd\sigma_2 + (ad+bc)z_2>0$ for
$r>0$. Similarly the term in the last two lines of \refb{ereorg} 
may be identified as representing the subtraction of a bound state of the
half BPS states carrying charges $(aM, bM)$ and $(cN,dN)$, with
\be
M^2=2p, \quad N^2=2q, \quad M\cdot N=-r, \quad r>0\, ,
\ee
so that the total charge $(Q,P)=(aM+cN,bM+dN)$
satisfies
\be
Q^2 = 2 (pa^2+qc^2-r ac), \quad P^2 = 2(pb^2+qd^2-r bd), \quad Q\cdot P =(2pab+2qcd - r(ad+bc)\, .
\ee
Such bound states are known to exist in the chamber $ac\tau_2 + bd\sigma_2 + (ad+bc)z_2<0$ for
$r>0$. Thus these subtraction terms remove from the Fourier expansion of $1/\Phi_{10}$
the contributions from two-centered bound states carrying such charges in any chamber of the
moduli space. Hence for these terms, \eqref{e216xy} gives the correct subtraction.

The terms involving one or two powers of $f_{-1}$ require special attention. Let us first consider the term proportional to $f_p f_{-1}$ with $p\ge 0$ and $f_{-1}f_q$ term with $q\geq 0$. In this case, 
following the same steps that led to \refb{ereorg}, 
the subtraction terms in \eqref{eguess} can be organized as:
\begin{eqnarray} \label{ereorgtwo}
&& - 
\sum_{\big{(}\begin{smallmatrix} a & b\cr c & d\end{smallmatrix}\big{)}\in \mr{PSL}(2,\mathbb{Z})} 
\sum_{r>0} \sum_{p\ge 0} f_p f_{-1}r  H(ac\tau_2 + bd\sigma_2 + (ad+bc)z_2) \,
 \nonumber \\ && \hskip 1in \times 
e^{2\pi i \{(pa^2-c^2+r ac)\tau+(pb^2-d^2+r bd)\sigma+
(2pab-2cd + r(ad+bc) )z\}}  \nonumber \\ &&
- 
\sum_{\big{(}\begin{smallmatrix} a & b\cr c & d\end{smallmatrix}\big{)}\in \mr{PSL}(2,\mathbb{Z})}
\sum_{r>0} \sum_{p\ge 0}f_p f_{-1}r  H(-ac\tau_2 - bd\sigma_2 - (ad+bc)z_2) \,
 \nonumber \\ && \hskip 1in \times 
e^{2\pi i \{(pa^2-c^2-r ac)\tau+(pb^2-d^2-r bd)\sigma+
(2pab-2cd - r(ad+bc) )z\}}   \, .
\end{eqnarray}
Note that we have used the $(a,b,c,d)\to (c,d,-a,-b)$, $p\leftrightarrow q$ 
symmetry to remove the overall
factor of $1/2$ and drop the $f_{-1}f_q$ term.
The terms in \eqref{ereorgtwo} have the same interpretation as in the case of the
terms proportional to $f_p f_{q}$ with $q$ replaced by $-1$, 
namely they represent subtraction of the contribution
from the bound states of charges
$(aM, cM)$ and $(cN,dN)$, with
\be
M^2=2p, \quad N^2=-2, \quad M\cdot N=\pm r, \quad r>0\, ,
\ee
so that the total charge $(Q,P)=(aM+cN,bM+dN)$
satisfies
\be\label{ebcount}
Q^2 = 2 (pa^2-c^2\pm r ac), \quad P^2 = 2(pb^2-d^2\pm r bd), \quad Q\cdot P =(2pab-2cd \pm r(ad+bc)\, .
\ee
However, it was shown in \cite{1104.1498,1210.4385} that 
this overcounts the bound state contribution. To see this, let us
make a change of variable $(a,b,c,d)\to (a,b,c-ra,d-rb)$ to express 
the term given in the third 
and fourth line of \refb{ereorgtwo} as
\begin{eqnarray} \label{emetamor}
&& - 
\sum_{\big{(}\begin{smallmatrix} a & b\cr c & d\end{smallmatrix}\big{)}\in \mr{PSL}(2,\mathbb{Z})}
\sum_{r>0} \sum_{p\ge 0}f_p f_{-1} r  H\left(-ac\tau_2 - bd\sigma_2 - (ad+bc)z_2 + ra^2\tau_2 + rb^2\sigma_2+2rab z_2
\right) \,
\nonumber \\ && \hskip 1in \times 
e^{2\pi i \{(pa^2-c^2+r ac)\tau+(pb^2-d^2+r bd)\sigma+
(2pab-2cd + r(ad+bc)) z\}}\, .
\end{eqnarray}
We now note that the exponent in the summand has exactly the same form as the exponent in the
the second line of \refb{ereorgtwo} and hence represent the same charges.
Naively these two contributions will get added. However, it was shown in 
\cite{1104.1498,1210.4385} that these two contributions should be regarded as coming
from the same bound states and 
furthermore, for given $(\tau_2,\sigma_2,z_2)$, the bound state 
exists only when the Heaviside functions appearing in \eqref{emetamor} and the first line of \eqref{ereorgtwo} are both non-zero. 
Thus \refb{ereorgtwo} should be
replaced by,
\begin{eqnarray} \label{erewritetwo}
&& - 
\sum_{\big{(}\begin{smallmatrix} a & b\cr c & d\end{smallmatrix}\big{)}\in \mr{PSL}(2,\mathbb{Z})}
\sum_{r>0} \sum_{p\ge 0} f_p f_{-1} r  H(ac\tau_2 + bd\sigma_2 + (ad+bc)z_2) \nonumber \\ && \hskip 1in \times\
H\left(-ac\tau_2 - bd\sigma_2 - (ad+bc)z_2 + ra^2\tau_2 + rb^2\sigma_2+2rab z_2
\right) \,
\nonumber \\ && \hskip 1in \times 
e^{2\pi i \{(pa^2-c^2+r ac)\tau+(pb^2-d^2+r bd)\sigma+
(2pab-2cd + r(ad+bc) )z\}} 
 \, .
\end{eqnarray}

We can verify that in this form, the divergence of the type mentioned in  \refb{e314}
disappears. Indeed, if we set $(a,b,c,d)=(1,0,c,1)$, the sum in 
\refb{erewritetwo} takes the form
\be \label{erewritecheck}
- 
\sum_{c\in \IZ} 
\sum_{r>0} \sum_{p\ge 0} r   \,
f_p f_{-1} e^{2\pi i \{(p-c^2+r c)\tau-\sigma+
(-2c + r)z\}} 
H(c\tau_2  + z_2)  H\left(-c\tau_2 -z_2 +
r\tau_2 \right)
\, . 
\ee
We now see that for fixed $r$, the sum over $c$ is restricted, and hence the
sum over $c$ no longer diverges.
In Section \ref{stildeFconverge} we shall prove the convergence of the 
full sum over $a,b,c,d$, $r$ and $p$.

Finally, we turn to the terms proportional to $f_{-1}^2$. 
We start from 
\begin{eqnarray} \label{ef1analysis}
&& -{1\over 2} f_{-1}^2
{\sum_{\big{(}\begin{smallmatrix} a & b\cr c & d\end{smallmatrix}\big{)}\in \mr{PSL}(2,\mathbb{Z})}}
\sum_{r>0} r  H(ac\tau_2 + bd\sigma_2 + (ad+bc)z_2) \,
 \nonumber \\ && \hskip 1in \times 
e^{2\pi i \{(-a^2-c^2+r ac)\tau+(-b^2-d^2+r bd)\sigma+
(-2ab-2cd + r(ad+bc) )z\}}  \nonumber \\ &&
-{1\over 2} f_{-1}^2
{\sum_{\big{(}\begin{smallmatrix} a & b\cr c & d\end{smallmatrix}\big{)}\in \mr{PSL}(2,\mathbb{Z})}}
\sum_{r>0} r  H(-ac\tau_2 - bd\sigma_2 - (ad+bc)z_2) \,
\nonumber \\ && \hskip 1in \times 
e^{2\pi i \{(-a^2-c^2-r ac)\tau+(-b^2-d^2-r bd)\sigma+
(-2ab-2cd - r(ad+bc) )z\}} \, .
\end{eqnarray}
We now note that the trasformation $(a,b,c,d)\to (-c,-d,a,b)$ exchanges the two terms. 
Therefore, we can drop the factor of $1/2$ and drop the second term in the sum. This gives:
\begin{eqnarray} \label{ef1analysistwo}
&& - f_{-1}^2
{\sum_{\big{(}\begin{smallmatrix} a & b\cr c & d\end{smallmatrix}\big{)}\in \mr{PSL}(2,\mathbb{Z})}}
\sum_{r>0} r\,  H(ac\tau_2 + bd\sigma_2 + (ad+bc)z_2) \,
 \nonumber \\ && \hskip 1in \times 
e^{2\pi i \{(-a^2-c^2+r ac)\tau+(-b^2-d^2+r bd)\sigma+
(-2ab-2cd + r(ad+bc) )z\}}   \, .
\end{eqnarray}
In this new representation the transformation $(a,b,c,d)\to (a,b,c-ra,d-rb)$ described above \refb{emetamor} is no longer available since that related the terms we have dropped to the terms we have kept. However, we can combine the transformations $(a,b,c,d)\to (-c,-d,a,b)$ and  $(a,b,c,d)\to (a,b,c-ra,d-rb)$ as 
\be \label{ef1analysisthree}
\begin{pmatrix} a & b\cr c & d\end{pmatrix} \to  \begin{pmatrix} 0 & -1 \cr 1 & 0\end{pmatrix} \begin{pmatrix} 1 & -r \cr 0 & 1\end{pmatrix}\begin{pmatrix} a & b\cr c & d\end{pmatrix}= \begin{pmatrix} 0 & -1 \cr 1 & -r\end{pmatrix}\begin{pmatrix} a & b\cr c & d\end{pmatrix}~.
\ee
This keeps the summand in \refb{ef1analysistwo} unchanged.
Due to the symmetry between $\sigma$ and $\tau$, we also have another symmetry:
\be \label{ef1analysisfour}
\begin{pmatrix} a & b\cr c & d\end{pmatrix} \to \begin{pmatrix} 0 & -1 \cr 1 & 0\end{pmatrix} \begin{pmatrix} 1 & 0 \cr -r & 1\end{pmatrix}  \begin{pmatrix} a & b\cr c & d\end{pmatrix}= \begin{pmatrix} r & -1 \cr 1 & 0\end{pmatrix}\begin{pmatrix} a & b\cr c & d\end{pmatrix}\, .
\ee
By the same argument as before, the bound states corresponding to 
matrices $(\begin{smallmatrix} a & b\cr c & d\end{smallmatrix})$ related by 
these transformations represent the same physical state and
should be counted only once\cite{1104.1498,1210.4385}. Thus in the sum over $a,b,c,d$
in \refb{ef1analysistwo}, we must identify the matrices 
$(\begin{smallmatrix} a & b\cr c & d\end{smallmatrix})$ related by these transformations. 
Since the matrices given in \refb{ef1analysisthree} and \refb{ef1analysisfour} are inverses of each other up to an overall multiplication by $-1$, we need to sum over all $\mr{PSL}(2,\mathbb{Z})$ matrices up to left multiplication by the group $G_r$ generated by the matrix $(\begin{smallmatrix} 0 & -1\\1  & -r\end{smallmatrix})$. Furthermore, 
these bound states exist in only those chambers of the moduli space where the 
Heaviside function in \refb{ef1analysistwo}, and all the other Heaviside functions
related to the one in \refb{ef1analysistwo} by replacing the matrix 
$(\begin{smallmatrix} a & b\cr c & d\end{smallmatrix})$ by an element of $G_r$
multiplying this matrix from the left, are non-zero. Thus
the
the Heaviside function in \refb{ef1analysistwo} must be replaced by the product of infinite number of Heaviside functions where in  the argument we replace $(a,b,c,d)$ by $(a_n,b_n, c_n, d_n)$, where
\be\label{edefanbn}
\begin{pmatrix} a_n & b_n\cr c_n & d_n\end{pmatrix}:=\begin{pmatrix} 0 & -1\cr 1 & -r\end{pmatrix}^n
\begin{pmatrix} a & b\cr c & d\end{pmatrix}\,,\quad n\in\IZ .
\ee
Thus the net contribution may be expressed as
\begin{eqnarray} \label{ereorgthreen}
&-&  f_{-1}^2 \sum_{r>0} r 
\sum_{\big{(}\begin{smallmatrix} a & b\cr c & d\end{smallmatrix}\big{)}\in G_r\backslash \mr{PSL}(2,\mathbb{Z})}\hskip .1in 
 \bigg\{\prod_{n=-\infty}^\infty 
H(a_nc_n\tau_2 + b_nd_n\sigma_2 + (a_nd_n+b_nc_n)z_2)
\bigg\}
\nonumber \\ && \hskip 1in \times \
e^{2\pi i \{(-a^2-c^2+r ac)\tau+(-b^2-d^2+r bd)\sigma+
(-2ab-2cd + r(ad+bc)) z\}}  \, .
\end{eqnarray}
The sum is well defined since the summand is invariant under a change in the coset representative. 
Note that since $G_r$ depends on $r$, we are forced to exchange the order of sum over $r$ and sum over $a,b,c,d$. This is part of the prescription in the definition of $\wt{F}(\Omega)$.

Using \refb{e216xy}, \refb{erewritetwo} and \refb{ereorgthreen}, 
we obtain 
\begin{eqnarray}
\label{eguessfin}
\wt F(\Omega) &=& {1\over \Phi_{10}(\Omega)}
- {1\over 2} \sum_{\big{(}\begin{smallmatrix} a & b\cr c & d\end{smallmatrix}\big{)}\in 
\mr{PSL}(2,\mathbb{Z})}
\left(e^{\pi i \{ac\tau + bd\sigma + (ad+bc)z\}} - e^{-\pi i \{ac\tau + bd\sigma + (ad+bc)z\}}\right)^{-2} \nonumber \\ && 
\hskip 1in 
\times\ f_+(a^2\tau +b^2\sigma +2abz) \ f_+(c^2\tau+d^2\sigma+2cd z) \nonumber \\ 
&-&  
\sum_{\big{(}\begin{smallmatrix} a & b\cr c & d\end{smallmatrix}\big{)}\in  \mr{PSL}(2,\mathbb{Z})}\sum_{r>0}  r\sum_{p\ge 0} f_p f_{-1}
 H(ac\tau_2 + bd\sigma_2 + (ad+bc)z_2) \nonumber \\ && \hskip 1in \times\
H\left(-ac\tau_2 - bd\sigma_2 - (ad+bc)z_2 + ra^2\tau_2 + rb^2\sigma_2+2rab z_2
\right) \,
\nonumber \\ && \hskip 1in \times 
e^{2\pi i \{(pa^2-c^2+r ac)\tau+(pb^2-d^2+r bd)\sigma+
(2pab-2cd + r(ad+bc) )z\}} 
\nonumber \\ 
&-&  f_{-1}^2 \sum_{r>0} r\, 
\sum_{\big{(}\begin{smallmatrix} a & b\cr c & d\end{smallmatrix}\big{)}\in G_r\backslash \mr{PSL}(2,\mathbb{Z})}\hskip .1in 
 \bigg\{\prod_{n=-\infty}^\infty 
H(a_nc_n\tau_2 + b_nd_n\sigma_2 + (a_nd_n+b_nc_n)z_2)
\bigg\}
\nonumber \\ && \hskip 1in \times \
e^{2\pi i \{(-a^2-c^2+r ac)\tau+(-b^2-d^2+r bd)\sigma+
(-2ab-2cd + r(ad+bc)) z\}}  \, .
\end{eqnarray}
The sum over $a,b,c,d$ will be organized by first summing over $a,b,c,d$ subject to the condition $|a|,|b|,|c|,|d|\leq K$ for some positive integer $K$ so that the sum is finite and then summing over all positive integers $K$.  
In Section \ref{stildeFconverge} we shall prove the convergence of the sum given in
\refb{eguessfin}.
\section{Analysis of the generating function $F(\Omega)$} \label{sgenerating}

Our goal in this section is to examine some properties of the series $F(\Omega)$ defined in \refb{edefFO}.

First we shall show that $d^*(m,n,\ell)$ defined in \refb{edstarTint} gives an unambiguous result 
for zero discriminant states if we follow the procedure described below \refb{eq 3.3}. For this 
we recall that the source of the ambiguity are the poles of $1/\Phi_{10}$ in the region where 
$\det\rm Im\, \Omega$ is 
larger than $1/4$. As discussed below \eqref{eq:im_re_poles_new}, in this region the poles are known to occur on the walls \refb{ewall}.
Furthermore, all such poles are known to be related by a $\mr{PSL}(2,\mathbb{Z})$-transformation (see Lemma \ref{lemma:sl2z_linear_poles}); so we
can focus on the pole corresponding to $m_1=n_1=0$, $j=1$. This corresponds to the
subspace
\be 
z_2=0\, .
\ee
Our goal will be to check if the existence of poles on this subspace could cause possible
ambiguity in the determination of $d^*(m,n,\ell)$ via \refb{edstarTint}.  Since
\refb{edeformedcontour} tells
us that for charge $(m,n,\ell)$ we take $z_2=-(\ell+\varepsilon_3)/\varepsilon$, 
the wall at $z_2=0$ can affect the computation of $d^*(m,n,\ell)$ only for $\ell=0$. Otherwise
the sign of $z_2$ will be determined by the sign of $\ell$.
On the other hand, we see from \eqref{epolestructure} that $1/\Phi_{10}$ has a double pole at $z=0$. Hence 
the jump in the integral in \refb{edstarTint} as we cross the subspace
$z_2=0$ is proportional to the $z$-derivative of the exponential factor in
\refb{edstarTint} that brings down a factor of $\ell$.
In particular, we see from \eqref{epolestructure}, \eqref{edeff} that the residue is proportional to \cite{Sen:2007qy}
\be \label{ejumpx}
\ell\, f_m \, f_n~.
\ee
Since \refb{ejumpx} vanishes for $\ell=0$ we see that the ambiguity in defining the
integration contour does not affect the computation of $d^*(m,n,\ell)$ for zero
discriminant states.

Note that the same argument is valid even for $4mn>\ell^2$. If $\ell=0$, the contour
$\mathcal{C}_{m,n,\ell}$ has $z_2=0$ and hence the $z=0$ pole of $1/\Phi_{10}$ lies on
the integration contour. But the residue at the pole is proportional to $\ell$ and
hence vanishes for $\ell=0$. So there is no ambiguity in computing $d^*(m,n,\ell)$.
$\mr{PSL}(2,\mathbb{Z})$ invariance of the expression then shows that there are also no ambiguities
due to other poles of $1/\Phi_{10}$.

We shall now examine the convergence of the sum appearing in \refb{edefFO}.
We get from \refb{edefFO}:
\be \label{e25x}
|F(\Omega)| \leq  \sum_{m,n,\ell\in {\mathbb{Z}}} |d^*(m,n,\ell)|\, 
e^{-2 \pi \left(m\tau_2+n\sigma_2+\ell z_2\right)}
\le \sum_{m,n,\ell\in {\mathbb{Z}}\atop m,n\ge 0} |d^*(m,n,\ell)| \, e^{-2\pi \left(
m\tau_2+n\sigma_2-|\ell|  |z_2|\right)}\, ,
\ee
where we used the fact that $d^*(m,n,\ell)$ is non-zero only for $m,n\ge 0$.
Now when all components of $T$ grow together, we have \cite{Sen:2007qy}
\begin{equation} \label{eboundtotal}
|d^*(T)| < C\, e^{2\pi\sqrt{\mr{det}(T)}} = C\, e^{\pi \sqrt{4mn-\ell^2}} ~,   
\ee
for some positive constant $C$.
This gives
\begin{equation} \label{e27x}
  |F(\Omega)|\leq C\sum_{m,n,\ell\in\IZ\atop m,n\ge 0}e^{-2\pi \left(
m\tau_2+n\sigma_2-|\ell|  |z_2|\right) + \pi \sqrt{4mn-\ell^2}}~.  
\end{equation}
In order to put a bound on the summand, we shall maximize the exponent
subject to the condition
\be  \label{e28xx}
m\tau_2 + n\sigma_2 = A\, ,
\ee
where $A$ is a constant. 
This can be done using a
Lagrange multiplier. Extremization with respect to $m,n,\ell$ leads to the
equations
\ben
&&-2\pi \tau_2(1+\lambda) + \pi \, {2n \over \sqrt{4mn-\ell^2}} =0~, \nonumber \\ &&
-2\pi \sigma_2(1+\lambda) + \pi \, {2m \over \sqrt{4mn-\ell^2}} =0~, \nonumber \\ &&
2\pi |z_2| -\pi \, {|\ell|\over \sqrt{4mn-\ell^2}}=0\, ,
\een
where $\lambda$ is the Lagrange multiplier.
These equations, together with \refb{e28xx}, give
\be
m = {A\over 2\tau_2}, \quad n = {A\over 2\sigma_2}, \quad \ell^2 = {16mn z_2^2
\over 1 + 4z_2^2}= {4\, A^2\, z_2^2\over \sigma_2\tau_2(1 + 4z_2^2)}\, .
\ee
Substituting this into \refb{e27x} and using the inequality
\be
(\tau_2\sigma_2 - z_2^2) \ge {1\over 4} (1+\epsilon)\, , 
\ee
for some small but positive real number $\eps$, we get
\be
|F(\Omega)| \le C \sum_{m,n,\ell\in\IZ\atop m,n\ge 0}
e^{-2\pi A \left(1 - \sqrt{1 - {\epsilon\over 4\tau_2\sigma_2}}\right)}
= C \sum_{m,n,\ell\in\IZ\atop m,n\ge 0}
e^{-2\pi (m\tau_2+n\sigma_2) \left(1 - \sqrt{1 - {\epsilon\over 4\tau_2\sigma_2}}\right)}\, .
\ee
Now for given $m,n$, the number of possible values of $\ell$ subject to the condition 
$4mn-\ell^2\ge 0$ is bounded from above
by $(2\sqrt{4mn}+1)$. This gives
\be
|F(\Omega)| \le  C  \sum_{m,n\in \mathbb{Z}\atop m,n\ge 0}\, (2\sqrt{4mn}+1) \, 
e^{-2\pi (m\tau_2+n\sigma_2) 
\left(1 - \sqrt{1 - {\epsilon\over 4\tau_2\sigma_2}}\right)} <\infty\, .
\ee
This proves the convergence of $F(\Omega)$ from the region where $m,n,\ell$ grow together. 

The growth formula for $d^*(m,n,\ell)$ takes a different form in other directions when only a subset of the charges
$(m,n,\ell)$ become large. For example we can consider the Cardy-like 
limit where only $m$ becomes large
keeping fixed $n$ and $\ell$. In such cases the index grows as $ m^\alpha 
\exp[c\sqrt{m}]$ for constants $\alpha$ and $c$. For given $m\ge n$ and 
$\det T\ge 0$ we have $|\ell|\le 4mn$ and 
the number of matrices $T$
grows polynomially with $m$. Presence of the $e^{-2\pi m\tau_2}$ factor in
\refb{e25x} now shows that the sum over $m$ converges.

We have not carefully examined the convergence of the sum in all possible directions\footnote{This may be possible with the help of a more detailed
results given in \cite{LopesCardoso:2021aem}.}, but the
analysis described above 
proves that if the bound \refb{eboundtotal} on $d^*(T)$ holds for all charges, and in
particular $d^*(T)$ vanishes for $4mn-\ell^2<0$, then
the sum converges for 
$\det\mr{Im}\,\Omega>1/4$.
We can also relax the bound \refb{eboundtotal} by multiplying the right-hand side
by finite degree polynomials in $m,n,\ell$ without affecting the convergence
property.
In Section \ref{ssection5} 
we shall show the equality of $F(\Omega)$ and $\wt F(\Omega)$ and also show that
$\wt F(\Omega)$ is analytic in the domain $\det\mr{Im}\,\Omega>1/4$.
This would prove convergence of \refb{edefFO}.

\section{Convergence of $\wt F(\Omega)$} \label{stildeFconverge}

In this section we shall prove the convergence of the sum in \eqref{eguessfin}. 
Let us define 
\begin{equation} \label{ef1f2f3}
\begin{split}
    \CF_{1}&:=\sum_{\big{(}\begin{smallmatrix} a & b\cr c & d\end{smallmatrix}\big{)}\in \mr{PSL}(2,\mathbb{Z})} 
\left(e^{\pi i \{ac\tau + bd\sigma + (ad+bc)z\}} - e^{-\pi i \{ac\tau + bd\sigma + (ad+bc)z\}}\right)^{-2} \\ &\hskip 1in 
\times\ f_+(a^2\tau +b^2\sigma +2abz) \ f_+(c^2\tau+d^2\sigma+2cd z)~,
\\
\CF_2&:=
\sum_{\big{(}\begin{smallmatrix} a & b\cr c & d\end{smallmatrix}\big{)}\in \mr{PSL}(2,\mathbb{Z})}\sum_{r>0} r \sum_{p\ge 0} f_p f_{-1} 
 H(ac\tau_2 + bd\sigma_2 + (ad+bc)z_2)
 \\ & \hskip 1in \times\
H\left(-ac\tau_2 - bd\sigma_2 - (ad+bc)z_2 + ra^2\tau_2 + rb^2\sigma_2+2rab z_2
\right)\\ & \hskip 1in \times 
e^{2\pi i \{(pa^2-c^2+r ac)\tau+(pb^2-d^2+r bd)\sigma+
(2pab-2cd + r(ad+bc) )z\}}~, 
\\
\CF_3&:=  f_{-1}^2 \sum_{r>0} r 
\sum_{\big{(}\begin{smallmatrix} a & b\cr c & d\end{smallmatrix}\big{)}\in G_r\backslash \mr{PSL}(2,\mathbb{Z})}\hskip .1in 
 \bigg\{\prod_{n=-\infty}^\infty 
H(a_nc_n\tau_2 + b_nd_n\sigma_2 + (a_nd_n+b_nc_n)z_2)
\bigg\}
\\ & \hskip 1in \times \
e^{2\pi i \{(-a^2-c^2+r ac)\tau+(-b^2-d^2+r bd)\sigma+
(-2ab-2cd + r(ad+bc)) z\}}  \, ,
\end{split}    
\end{equation}
so that \eqref{eguessfin} can be written as
\begin{equation}\label{eq:tildeFCF123}
    \widetilde{F}(\Omega)=\frac{1}{\Phi_{10}}-\frac{1}{2}\CF_1-\CF_2-\CF_3~.
\end{equation}



\begin{figure}
    \centering
    \begin{tikzpicture}[scale=3]

  \draw[->] (-1.6,0) -- (1.6,0) node[right] {$x:=\frac{z_2}{\tau_2}$};
  \draw[->] (0,0) -- (0,2.2) node[above] {$y:=\frac{\sigma_2}{\tau_2}$};

  \draw[thick,domain=-1.5:1.5,samples=200] plot (\x,{\x*\x}) node[right]{} ;

  \draw[thick] (0,0) -- (-1,1);
  \draw[thick] (0,0) -- (1,1);

  \draw[thick] (-1,1) -- (-1,2);
  \draw[thick] (1,1) -- (1,2);
  \draw[thick] (0,0) -- (0.5,0.25);
  \draw[thick] (0.5,0.25) -- (1,1);
\draw[thick] (0,0) -- (0,2);
  \node at (-0.5,1.4) {$\mathcal{R}$};
  \node at (0.5,1.4) {$\mathcal{L}$};
  \node at (-1.1,1) {$-1$};
  \node at (1.1,1) {$1$};
\node at (0,-0.1) {$0$};
  \node at (0.65,0.25) {$\frac{1}{2}$};
\end{tikzpicture}
    \caption{Chambers in the $x$-$y$ plane separated by pole locations of $1/\Phi_{10}$.}
    \label{figfour}
\end{figure}
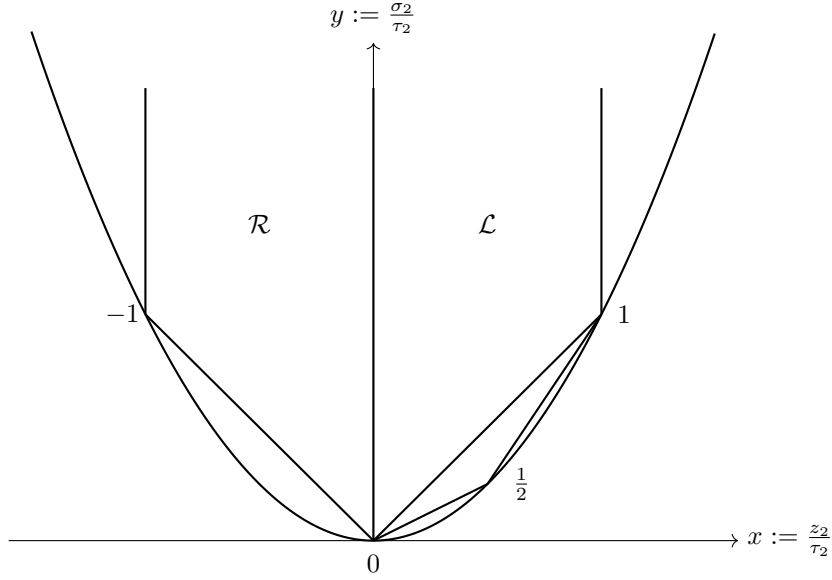

For this analysis, it will be useful to develop a geometric picture for the matrices
$\begin{pmatrix} a & b\cr c & d\end{pmatrix}\in \mr{SL}(2,\mathbb{Z})$. For this, let us define
\be
x := {z_2\over \tau_2}, \qquad y := {\sigma_2\over \tau_2}\, .
\ee
Then the imaginary parts of \refb{epoles} for $n_2=0$ take the form:
\be \label{epolesimaginary}
- m_1 +n_1\, y + j\, x
=0, \quad
m_1,n_1,j\in \IZ,
\quad m_1 n_1 + {j^2\over 4}={1\over 4}\, .
\ee
On the other hand, the condition $\sigma_2\tau_2>z_2^2$ takes the form
\be\label{eparabola}
y > x^2\, .
\ee
\refb{eparabola} gives a region in the $x$-$y$ plane bounded below by a parabola and the straight lines represented by 
\refb{epolesimaginary} divide this region into infinite number of chambers. One can show that the lines represented by \eqref{epolesimaginary} intersect only on the parabola and hence the vertices 
of these chambers either lie on the parabola or at $y=\infty$. The vertices
lying on the parabola will be labelled by their $x$ values and the vertex at $y=\infty$ will
be denoted by $\infty$. Using this notation we can specify a chamber by its vertices. For
example, the chamber labeled as $\CR$ in Figure \ref{figfour}
will be labeled as $(-1,0,\infty)$ and the chamber
labeled as $\CL$ in Figure \ref{figfour} is labeled as $(0,1,\infty)$. This figure contains
another chamber $(0,1/2,1)$, bounded by the lines $x=y$, $x=2y$ and $3x=2y+1$.

Consider now the $\mr{SL}(2,\mathbb{Z})$ transformation 
\begin{equation} \label{edefprime}
    \tau'_2:=a^2\tau_2+b^2\sigma_2+2abz_2,\quad \sigma'_2:=c^2\tau_2+d^2\sigma_2+2cdz_2,\quad z_2':=ac\tau_2+bd\sigma_2+(ad+bc)z_2~.
\end{equation}
Under this, $x$ and $y$
transform as
\be\label{esl2zonxy}
x\mapsto x'= {ac+bd y+(ad+bc)x\over a^2+b^2 y+2ab x}, \qquad
y\mapsto y'={c^2+d^2 y+2cd x\over a^2+b^2 y+2ab x}
\ee
This gives
\be
- m_1' +n_1'\, y' + j'\, x' = {-m_1+ n_1 y + j x\over a^2+b^2 y+2ab x}\, ,
\ee
with
\be
m_1 =m_1'a^2 - n_1' c^2  - j'ac , \qquad n_1 = - m_1' b^2 + n_1'd^2 + j' bd,
\qquad j = - 2 m_1' ab + 2 n_1' cd + (ad+bc) j'\, .
\ee
One can show that 
\begin{equation}
m_1 n_1+{j^2\over 4}=(ad-bc)^2\left(m_1'n_1'+\frac{j'^2}{4}\right)=m_1'n_1'+\frac{j'^2}{4}~.    
\end{equation}
Thus
\be
m_1' n_1'+{j^{\prime 2}\over 4} = \frac{1}{4} \quad \iff \quad m_1n_1+{j^2\over 4}=\frac{1}{4}\, ,
\ee
and as a result the chambers themselves are mapped
to each other under the transformation \refb{esl2zonxy}. The map can be read out from the map
between the vertices of the chamber. For example the chamber $\CR=(-1,0,\infty)$ is mapped to
\be \label{eCRmap}
\left({c-d\over a-b}, {c\over a}, {d\over b}\right)\, ,
\ee
under \refb{esl2zonxy}.

Since in our expression for $\CF_{i},~i=1,2,3$, replacing $\Omega$ by $\gamma\Omega
\gamma^t$  has the effect of multiplying the $\mr{PSL}(2,\mathbb{Z})$ matrix
$\begin{pmatrix} a&b\cr c&d\end{pmatrix}$ by $\gamma$ from the right, and we
are summing over all $\mr{PSL}(2,\mathbb{Z})$ matrices,
$\CF_{i},~i=1,2,3$ are formally invariant under the transformation
\refb{edefprime}:
\be \label{esl3zoffi}
\CF_i(\gamma\Omega\gamma^t)=\CF_i(\Omega), \qquad\gamma\in\mr{PSL}(2,\IZ)~,\quad i=1,2,3\, .
\ee
This equation is formal since we have not yet proved the convergence of the sum over
$a,b,c,d$. However,
once we prove convergence by restricting $(\tau_2,\sigma_2,z_2)$
in a particular chamber, we can use \refb{esl3zoffi} to find a finite result for
$\CF_i(\gamma\Omega\gamma^t)$ in any other chamber.
Hence we can restrict the original choice of $(x,y)$ to one particular
chamber for proving convergence of the sum. We shall take this to be the 
chamber $\CR$. 
We will prove the following theorem:
\begin{thm}\label{thm:S_conv}
The sum over $a,b,c,d$ and $r$ in \eqref{eguessfin} converges absolutely and uniformly on compact subsets of $\CR$. 
\end{thm}
We will prove this by showing that the sums $\CF_1,\CF_2,\CF_3$ converges absolutely and uniformly on compact subsets of $\CR$. This will be done in a series of lemmas and propositions. 
\subsection{Convergence of $\CF_1$}\label{sec:conv_CF1}
As discussed above, we shall work in the chamber
\be\label{echamber}
\CR:\qquad z_2<0, \qquad z_2>-\tau_2, \qquad z_2>-\sigma_2\, .
\ee
Let us define the series $\CF_{00}$:
\begin{eqsp}\label{eq:CFij_def}
\CF_{00}&:=
\sum_{\big{(}\begin{smallmatrix} a & b\cr c & d\end{smallmatrix}\big{)}\in \mr{PSL}(2,\mathbb{Z})}
\left(e^{\pi i \{ac\tau + bd\sigma + (ad+bc)z\}} - e^{-\pi i \{ac\tau + bd\sigma + (ad+bc)z\}}\right)^{-2}\, .    
\end{eqsp}

\begin{prop}\label{prop:conv_F00}
The series $\CF_{00}$ converges absolutely and uniformly on compact subsets of the $\CR$-chamber.    
\end{prop}
\begin{proof}
We write,
\be\label{exx2}
ac\tau_2 + bd\sigma_2 + (ad+bc)z_2 = ac (\tau_2+z_2) + bd (\sigma_2 + z_2)
+ (a-b) (c-d)) (-z_2)\, .
\ee
Now for $\Omega\in \CR$, we have
\be\label{exx3}
\tau_2+z_2>0, \qquad \sigma_2+z_2> 0, \qquad -z_2>0\, .
\ee
Also using the condition $ad-bc=1$ and $a,b,c,d\in\mathbb{Z}$ one can show that
$ad$ and $bc$ have the same signs and hence
$ac$ and $bd$ have the same signs.  This includes the case where one of them vanishes
in which case we declare them to have the same sign. Writing the same condition as
$(a-b)d -  b(c-d)=1$ we see that $(a-b) (c-d)$ has the same sign as $bd$ and writing
the same condition as $(d-c)a-(b-a)c=1$ we conclude that $(d-c)(b-a)$ has the
same sign as $ac$. Therefore $ac$, $bd$ and $(a-b) (c-d)$ have the same signs.
Furthermore, while 
some of them may vanish, at
least one of the three combinations is non-zero. Let us now define
\be\label{exx4}
C(\Omega) := {\rm Min} (\tau_2+z_2, \sigma_2+z_2, -z_2), \qquad
\mu(a,b,c,d) := {\rm Max} (|a|, |b|, |c|, |d|)\, .
\ee
Then 
\be\label{e419xy}
{\rm Max} (|ac|, |bd|, |(a-b) (c-d)|)\ge \mu(a,b,c,d)\, ,
\ee
and
\refb{exx2} gives
\be \label{ecrucialinequality}
|ac\tau_2 + bd\sigma_2 + (ad+bc)z_2 | \ge C(\Omega)\, \mu(a,b,c,d)\, .
\ee
Now for any complex number $x=x_1+ix_2$ with real $x_1,x_2$, we
have
\begin{eqsp} \label{einequalityx}
 \left| (e^{ix}-e^{-ix})^{-2}\right|
  \le (e^{|x_2|}-e^{-|x_2|})^{-2} \le  (e^{x_0}-e^{-x_0})^{-2}\quad
  {\rm for} \ |x_2|\ge x_0 >0~.
\end{eqsp}
Using this we can put a bound on the summand in the expression for $\CF_{00}$:
\be \label{exx1}
\left|e^{\pi i \{ac\tau + bd\sigma + (ad+bc)z\}} - e^{-\pi i \{ac\tau + bd\sigma + (ad+bc)z\}}\right|^{-2}\le  \left(e^{\pi  C(\Omega) \mu(a,b,c,d)} - e^{-\pi C(\Omega) \mu(a,b,c,d)}\right)^{-2}\, .
\ee
Now, since $C(\Omega)$ is a continuous function of $\tau_2,\sigma_2,z_2$, it is bounded from below on any compact subset $\CO\subset\CR$. Let $C_{\text{lower}}(\CO)$ be the lower bound on the compact subset $\CO$:
\begin{eqsp}\label{eq:C_omega_lbound}
    C(\Omega)\geq C_{\text{lower}}(\CO)>0~,\quad \Omega\in\CO~.
\end{eqsp}
Finally, for a fixed positive integer $K$, the number of $\mr{PSL}(2,\IZ)$ matrices with $\mu(a,b,c,d)=K$ is bounded above by $4(2K+1)^2$. To see this, note that since the maximum element, which can be any of the four elements, is fixed at $K$, two of the other three elements can take atmost $2K+1$ values, and the fourth element is fixed by the $ad-bc=1$ condition.\footnote{This
analysis ignores the constraint that the fourth element must be an integer. Once we impose this
additional condition, the actual number of terms becomes less. However we shall proceed without
imposing this additional integrality condition, since the generous bound $4(2K+1)^2$ is
sufficient for our analysis.}
This gives
\begin{eqsp} \label{e423xxy}
|\CF_{00}| &\le \sum_{K=1}^\infty 4(2K+1)^2 \,  \left(e^{\pi  C(\Omega)K} - e^{-\pi 
C(\Omega) K}\right)^{-2}\\&\leq \sum_{K=1}^\infty 4(2K+1)^2 \,  \left(e^{\pi  C_{\text{lower}}(\CO)K} - e^{-\pi 
C_{\text{lower}}(\CO) K}\right)^{-2}
<\infty\, .
\end{eqsp}
This proves the absolute and uniform convergence of $\CF_{00}$ on compact subsets of $\CR$ and hence defines a holomorphic function on $\CR$.
\end{proof}

\begin{prop}
The series ${\cal F}_{1}$ converges absolutely and uniformly on compact subsets of the $\CR$-chamber.     
\end{prop}
\begin{proof}
First observe that in the $\CR$-chamber, we have
\begin{eqsp}\label{eq:bound_tau2}
   a^2\tau_2+b^2\sigma_2+2ab z_2=a^2(\tau_2+z_2)+b^2(\sigma_2+z_2)+(a-b)^2(-z_2)\geq C(\Omega)~,
   \\
   c^2\tau_2+d^2\sigma_2+2cd z_2=c^2(\tau_2+z_2)+d^2(\sigma_2+z_2)+(c-d)^2(-z_2)\geq C(\Omega)~,
\end{eqsp}
where $C(\Omega)$ is defined in \eqref{exx4}. This gives 
\begin{eqsp}\label{eq:f+_bound}
    |f_+(a^2\tau+b^2\sigma+2ab z)|\leq f_+(i(a^2\tau_2+b^2\sigma_2+2ab z_2))\leq f_+(i\, 
    C(\Omega))~,
    \\
    |f_+(c^2\tau+d^2\sigma+2cd z)|\leq f_+(i(c^2\tau_2+d^2\sigma_2+2cd z_2))
    \leq f_+(i\, C(\Omega))~,
\end{eqsp}
where we used $|f_+(\tau)| \le f_+(i\tau_2)$ and that $f_+(i\tau_2)$ is a monotonically
decreasing function of $\tau_2$ since $f_p>0$. Both these relations follow from \refb{e122}.
Now, since $f_+$ and $C(\Omega)$ are continuous functions, on any compact subset $\CO\subset\CR$, there exists $C_{\text{upper}}^{f_+}(\CO)>0$ such that 
\begin{eqsp}\label{eq:f+_bound_cpt_set}
    f_+(i\,C(\Omega))\leq C_{\text{upper}}^{f_+}(\CO)~,\quad \Omega\in\CO~.
\end{eqsp}
Using \refb{ef1f2f3}, \refb{eq:CFij_def},  \eqref{eq:f+_bound} and \eqref{eq:f+_bound_cpt_set} we have 
\begin{eqsp}
    |\CF_1(\Omega)|\leq \left(C_{\text{upper}}^{f_+}(\CO)\right)^2|\CF_{00}(\Omega)|~,\quad \Omega\in\CO~.
\end{eqsp}
Using Proposition \ref{prop:conv_F00}, we conclude that $\CF_1$ converges absolutely and uniformly on compact subsets of the $\CR$-chamber. 
\end{proof}
This proves the holomorphicity of $\CF_1$ in the $\CR$-chamber.

\subsection{Convergence of $\CF_2$}\label{sec:conv_CF2}
We now prove the convergence of $\CF_2$. The exponential term in $\CF_2$ can be rearranged to obtain 
\begin{equation}\label{eq:CF_2_def}
\begin{split}
    \CF_2&:=f_{-1}  
\sum_{\big{(}\begin{smallmatrix} a & b\cr c & d\end{smallmatrix}\big{)}
\in \mr{PSL}(2,\mathbb{Z})}\sum_{r>0} r\, 
 H(ac\tau_2 + bd\sigma_2 + (ad+bc)z_2)
 \\ & \hskip 1in \times\
H\left(-ac\tau_2 - bd\sigma_2 - (ad+bc)z_2 + ra^2\tau_2 + rb^2\sigma_2+2rab z_2
\right)\\ & \hskip 1in \times 
e^{2\pi i \{-(c^2\tau+d^2\sigma+2cdz)+r(ac\tau+bd\sigma+(ad+bc)z)\}}f_+(a^2\tau+b^2\sigma+2abz)~.
\end{split}    
\end{equation}
Note that the Heaviside functions do not jump in the interior of $\CR$. More precisely, the Heaviside functions are constant functions in the interior of the $\CR$-chamber taking value either 0 or 1. Hence, they are analytic function in the $\CR$-chamber. This implies that if we prove that the series converges absolutely and uniformly on compact subsets of $\CR$, then it defines a holomorphic function on $\CR$. We will now proceed to prove this.   \\\\ 
Let $\CO\subset\CR$ be a compact subset. Using $f_{-1}=1$ and the bound \eqref{eq:f+_bound} and \eqref{eq:f+_bound_cpt_set}, for $\Omega\in\CO$ we can write 
\begin{eqsp}\label{eq:bound_F20}
    |\CF_{2}|\leq C^{f_+}_{\mr{upper}}(\CO) \sum_{\big{(}\begin{smallmatrix} a & b\cr c & d
    \end{smallmatrix}\big{)}\in \mr{PSL}(2,\mathbb{Z})} 
\sum_{\substack{r>0}}& r\, H(ac\tau_2 + bd\sigma_2 + (ad+bc)z_2)
 \\ &\times\
H\left(-ac\tau_2 - bd\sigma_2 - (ad+bc)z_2 + ra^2\tau_2 + rb^2\sigma_2+2rab z_2
\right)\\ &\times 
e^{-2\pi  \{-(c^2\tau_2+d^2\sigma_2+2cdz_2)+r(ac\tau_2+bd\sigma_2+(ad+bc)z_2)\}}~. 
\end{eqsp}
We now note that $r$ in the sum is bounded from below because of the Heaviside functions. 
Put 
\begin{equation} \label{edefprimenew}
    \tau'_2:=a^2\tau_2+b^2\sigma_2+2abz_2,\quad \sigma'_2:=c^2\tau_2+d^2\sigma_2+2cdz_2,\quad z_2':=ac\tau_2+bd\sigma_2+(ad+bc)z_2~.
\end{equation}
The summand vanishes unless 
\begin{equation} \label{e361}
    z_2'>0,\quad r\tau_2'>z_2'\implies r>r_0:=z_2'/\tau_2'~.
\end{equation}
The sum then satisfies the bound 
\ben \label{esumbound}
   &&
    |\CF_{2}|\leq C^{f_+}_{\mr{upper}}(\CO)  \sum_{\big{(}\begin{smallmatrix} a & b\cr c & d\end{smallmatrix}
    \big{)}\in \mr{PSL}(2,\mathbb{Z})} 
\sum_{\substack{r\geq r_0\\r\in\IZ}} r \, H(z_2') \, e^{2\pi(\sigma_2'-rz_2')}~,\quad \Omega\in\CO~.
\een
Let $\r0$ be the lowest integer larger than or equal to $r_0$. 
Since $r_0=z_2'/\tau_2'$ and $z_2',\tau_2'>0$, $\r0$ must be a strictly positive integer. Then we can write the RHS of \refb{esumbound} as
\begin{equation} \label{esumboundnew}
C^{f_+}_{\mr{upper}}(\CO) \sum_{\big{(}\begin{smallmatrix} a & b\cr c & d\end{smallmatrix}
    \big{)}\in \mr{PSL}(2,\mathbb{Z})}
\sum_{\substack{r\geq \r0\\r\in\IZ}} r ~H(z_2') \, 
e^{2\pi (\tau_2\sigma_2 - z_2^2)/\tau_2'} \, 
e^{-2\pi \, (r-r_0) |z_2'|}
~,  
\end{equation}
where we used the fact that 
\begin{eqsp}
\tau_2\sigma_2 - z_2^2=\tau'_2\sigma'_2 - (z'_2)^2~.    
\end{eqsp}
Note that we have replaced $z_2'$ by $|z_2'|$ in the exponent due to the Heaviside function. 
Let us break the sum into $r=\r0$ and $r\geq \r0+1$ terms:
\begin{eqsp}
 |\CF_{2}|\leq     \CF_{2}^<+\CF_{2}^>~,
\end{eqsp}
where 
\begin{align}
\label{eq:CF20<_def}
    \CF_{2}^<&:= C^{f_+}_{\mr{upper}}(\CO)
    \sum_{\big{(}\begin{smallmatrix} a & b\cr c & d\end{smallmatrix}
    \big{)}\in \mr{PSL}(2,\mathbb{Z})}
\r0 ~H(z_2') \, 
e^{2\pi (\tau_2\sigma_2 - z_2^2)/\tau_2'} 
e^{-2\pi \, (\r0-r_0) |z_2'|}
\\
\CF_{2}^>&:= C^{f_+}_{\mr{upper}}(\CO)\sum_{\big{(}\begin{smallmatrix} a & b\cr c & d\end{smallmatrix}
    \big{)}\in \mr{PSL}(2,\mathbb{Z})}
\sum_{\substack{r\geq \r0+1\\r\in\IZ}} r ~H(z_2') \, 
e^{2\pi (\tau_2\sigma_2 - z_2^2)/\tau_2'} 
e^{-2\pi \, (r-r_0) |z_2'|}~.
\label{eq:CF20>0_def}
\end{align}
\begin{prop}\label{prop:conv_F20>}
The series $\CF_{2}^>$ converges absolutely and uniformly on compact subset of the $\CR$ chamber.    
\end{prop}
\begin{proof}
By \eqref{eq:bound_tau2} and \refb{ecrucialinequality}, we have 
\ben \label{ecomone}
   |z_2'| \ge \mu(a,b,c,d)\, C(\Omega)~,
\een
where $C(\Omega)$ and $\mu(a,b,c,d)$ have been defined in \refb{exx4}.
We also have, from \refb{edefprime} and \eqref{eq:bound_tau2}
\ben \label{ecomtwo}
&& \tau'_2=a^2(\tau_2+z_2)+b^2(\sigma_2+z_2)+ (a-b)^2 (-z_2) \ge C(\Omega), \\ &&   
\frac{2\pi (\tau_2\sigma_2 - z_2^2)}{\tau_2'}\leq\frac{2\pi (\tau_2\sigma_2 - z_2^2)}{C(\Omega)}, \nonumber\nonumber \\ &&
z_2'=ac( \tau_2+z_2) +bd(\sigma_2+z_2) +(a-b)(c-d)(-z_2) \le 6\, \mu(a,b,c,d)^2\, 
E(\Omega), \nonumber  \\ &&
E(\Omega) := {\rm max} (\tau_2+z_2, \sigma_2+z_2, -z_2)\, .
\een
This gives
\be
r_0=z_2'/\tau_2' \le 6 \,\mu(a,b,c,d)^2 E(\Omega) / C(\Omega)\, .
\ee
Furthermore, the number of $\mr{PSL}(2,\mathbb{Z})$ matrices with 
$\mu(a,b,c,d)=K$ is 
bounded from above by $4(2K+1)^2$. Wriring $s=r-\r0$, 
we get, from \refb{eq:CF20>0_def},
\be \label{e447x}
|\CF_{2}^>| \le C^{f_+}_{\mr{upper}}(\CO) \sum_{K=1}^\infty \sum_{s=1}^\infty 
 \, \left(s+1 + 6 K^2 \, {E(\Omega)\over C(\Omega)}\right)
4(2K+1)^2\, 
e^{2\pi (\tau_2\sigma_2 - z_2^2)/C(\Omega)} e^{-2\pi \, s\,  K\, C(\Omega)} 
<\infty\, ,
\ee
since $C(\Omega)>0$, $E(\Omega)>0$. 
This proves the absolute convergence of
$\CF_{2}^>$. The uniform and absolute convergence on compact subsets of $\CR$ follows from the continuity of $E(\Omega)$ and $C(\Omega)$.
\end{proof}

We now prove the convergence of $\CF_{2}^{<}$. 
Since $|z_2'|$ has a lower bound \eqref{ecomone} that grows when any of $|a|$, $|b|$, $|c|$, $|d|$ become large, we need to worry about the cases when $\r0-r_0$ becomes small in such limits to compensate for the growth of $|z_2'|$. If there are large number of possible $a,b,c,d$
satisfying this condition, then the sum over $a,b,c,d$ could cause divergence.
We shall now explore this possibility. 
To this end, choose $0<\eps<\frac{1}{2}$ and define 
\begin{eqsp}
    \CF_{2}^{<}=\CF_{2}^{<\eps}+\CF_{2}^{>\eps}~,
\end{eqsp}
where 
\begin{eqsp}\label{eq:F20<eps}
\CF_{2}^{<\eps}&:= C^{f_+}_{\mr{upper}}(\CO)\sum_{\big{(}\begin{smallmatrix} a & b\cr c & d\end{smallmatrix}
    \big{)}\in \mr{PSL}(2,\mathbb{Z})\atop r_0/\r0<\eps}
\r0 ~H(z_2') \, 
e^{2\pi (\tau_2\sigma_2 - z_2^2)/\tau_2'} 
e^{-2\pi \, (\r0-r_0) |z_2'|}~,
\end{eqsp}
\begin{eqsp}\label{eq:F20>eps}
\CF_{2}^{>\eps}&:= C^{f_+}_{\mr{upper}}(\CO) \sum_{\big{(}\begin{smallmatrix} a & b\cr c & d\end{smallmatrix}
    \big{)}\in \mr{PSL}(2,\mathbb{Z})\atop r_0/\r0\ge\eps} 
\r0 ~H(z_2') \, 
e^{2\pi (\tau_2\sigma_2 - z_2^2)/\tau_2'} 
e^{-2\pi \, (\r0-r_0) |z_2'|}~,    
\end{eqsp}
\begin{prop}
The series $\CF_{2}^{<\eps}$ converges absolutely and uniformly on compact subsets of the $\CR$-chamber. 
\end{prop}
\begin{proof}
Note that, if $\frac{r_0}{\r0}<\frac{1}{2}$, then $r_0<1$ and $\r0=1$. Then we have 
\begin{eqsp}
    \r0-r_0=1-r_0>1-\eps>0~.
\end{eqsp}
Thus, we have, following the argument leading to \refb{e447x}, 
\be \label{e454xxz}
\CF_{2}^{<\eps}\leq C^{f_+}_{\mr{upper}}(\CO)
\sum_{K=1}^\infty \, 
4(2K+1)^2\, 
e^{2\pi (\tau_2\sigma_2 - z_2^2)/C(\Omega)} e^{-2\pi \, (1-\eps) \,   K\, C(\Omega)} 
<\infty\, ,
\ee
This proves the convergence of $\CF_{2}^{<\eps}$. The uniform and absolute convergence on compact subsets of $\CR$ follows from the continuity of $C(\Omega)$.
\end{proof}

\begin{prop}
The series $\CF_{2}^{>\eps}$ converges absolutely and uniformly on compact subset of the $\CR$-chamber.     
\end{prop}
\begin{proof}
To find a lower bound on $(\r0-r_0)|z_2'|$, let us choose
the following parametrization of the $\mathrm{PSL}(2,\IZ)$-matrices: choose $k,\ell\in\IZ$ such that $c=a\,n-k,~d=b\,n-\ell$ where $n=\r0$. 
Then $ad-bc=1$ implies $-a\,\ell+b\,k=1$ and we have 
\begin{eqsp}
    r_0=\frac{ac\tau_2+bd\sigma_2+(ad+bc)z_2}{a^2\tau_2+b^2\sigma_2+2abz_2}=n-\frac{a\,k\tau_2+b\,\ell\sigma_2+(b\,k+a\,\ell)z_2}{a^2\tau_2+b^2\sigma_2+2abz_2}~.
\end{eqsp}
Thus $\r0=n$ if and only if 
\begin{eqsp}\label{eq:kl_cond_r0n}
0<\frac{a\,k\tau_2+b\,\ell\sigma_2+(b\,k+a\,\ell)z_2}{a^2\tau_2+b^2\sigma_2+2abz_2}<1~.    
\end{eqsp}
With this parametrization and $\frac{r_0}{\r0}\geq \eps$, we have 
\begin{eqsp}
(\r0-r_0) z_2' &
=(n-r_0) z_2' =r_0\tau_2'\left(n- r_0 
\right) \\&\geq  \eps\, n\,\tau_2'\left(\frac{a\,k\tau_2+b\,\ell\sigma_2+(b\,k
+a\,\ell)z_2}{a^2\tau_2+b^2\sigma_2+2abz_2}\right)\,
\\&=\eps\, n\, (a\,k\tau_2+b\,\ell\sigma_2+(b\,k+a\,\ell)z_2)~.    
\end{eqsp}
Thus we have
\begin{eqsp}\label{eq:ser_>eps}
\CF_{2}^{>\eps}&\leq C^{f_+}_{\mr{upper}}(\CO)\exp\left(\frac{2\pi (\tau_2\sigma_2 - z_2^2)}{C(\Omega)}\right)
\sum_{\big{(}\begin{smallmatrix} a & b\cr -k & -\ell\end{smallmatrix}
    \big{)}\in \mr{PSL}(2,\mathbb{Z})} \sum_{n=1}^{\infty}
n ~H(a\,k\tau_2+b\,\ell\sigma_2+(b\,k+a\,\ell)z_2) \, \\ & \hskip 2.5in
e^{-2\pi \, \eps\,n (a\,k\tau_2+b\,\ell\sigma_2+(b\,k+a\,\ell)z_2)}
\\&\leq C^{f_+}_{\mr{upper}}(\CO)\exp\left(\frac{2\pi (\tau_2\sigma_2 - z_2^2)}{C(\Omega)}\right)\\&\hspace{2cm}\times \sum_{\big{(}\begin{smallmatrix} a & b\cr -k & -\ell\end{smallmatrix}
    \big{)}\in \mr{PSL}(2,\mathbb{Z})}
\, \left(e^{\pi \, \eps\, (a\,k\tau_2+b\,\ell\sigma_2+(b\,k+a\,\ell)z_2)}-
e^{-\pi \, \eps\, (a\,k\tau_2+b\,\ell\sigma_2+(b\,k+a\,\ell)z_2)}\right)^{-2}~.
\end{eqsp}
The Heaviside function in the first line 
appears because of \eqref{eq:kl_cond_r0n}. Independent sums over $a,b,k,\ell,n$ clearly 
overestimates the original sum over $a,b,c,d$ but since we are trying to find an 
upper bound on $\CF_2^{>\eps}$, this is okay. 
The proof of convergence of last line in \eqref{eq:ser_>eps} is same as the proof of convergence of $\CF_{00}$ in Proposition \ref{prop:conv_F00}. 
\end{proof}

This proves the holomorphicity of $\CF_2$ in the $\CR$-chamber.

\subsection{Convergence of ${\cal F}_3$}\label{sec:conv_CF3}

We shall now analyze the convergence of 
${\cal F}_3$ given in \refb{ef1f2f3}:
\begin{equation} \label{ef3exp}
\begin{split}
\CF_3&:=  f_{-1}^2 \sum_{r>0} r 
\sum_{\big{(}\begin{smallmatrix} a & b\cr c & d\end{smallmatrix}\big{)}\in G_r\backslash \mathrm{PSL}(2,\mathbb{Z})}\hskip .1in 
 \bigg\{\prod_{n=-\infty}^\infty 
H(a_nc_n\tau_2 + b_nd_n\sigma_2 + (a_nd_n+b_nc_n)z_2)
\bigg\}
\\ & \hskip 1in \times \
e^{2\pi i \{(-a^2-c^2+r ac)\tau+(-b^2-d^2+r bd)\sigma+
(-2ab-2cd + r(ad+bc)) z\}}  \, .
\end{split}    
\end{equation}
The summand in this expression
has absolute value $e^{2\pi B}$ where
\be \label{eBdef}
B = \sigma_2'+\tau_2'-r\, z_2'\, ,
\ee
and $\tau_2',\sigma_2',z_2'$ have been defined in \refb{edefprimenew}. As discussed below \eqref{eq:CF_2_def}, the Heaviside functions are analytic in the interior of the $\CR$-chamber. Consequently, our goal is to show that the sum over $r$ and $a,b,c,d$ in the expression \refb{ef3exp} for ${\cal F}_3$ converges absolutely and uniformly on compact subsets of the $\CR$-chamber. 
We begin by analyzing the
constraints imposed by the Heaviside functions.

\begin{lemma}\label{lemma:Heaviside_matr_const}
Let $\big{(}\begin{smallmatrix} a & b\cr c & d\end{smallmatrix}\big{)}\in
\mathrm{PSL}(2,\IZ)$, $r\geq 1$ and $\big{(}\begin{smallmatrix} a_n & 
b_n\cr c_n & d_n\end{smallmatrix}\big{)}=\big{(}\begin{smallmatrix} 0 & -1\cr 1 & -r\end{smallmatrix}\big{)}^n\big{(}\begin{smallmatrix} a & b\cr c & d
\end{smallmatrix}\big{)}$. Let $a(n)$ be the sequence defined by the 
recursion relation\footnote{This sequence is an example of a Lucas sequence $U_n(r,1)$ \cite{Wikipedia:LucasSequence}. Pell numbers are another example of Lucas sequence which have featured in \cite{Chowdhury:2019mnb}.}:
\begin{eqsp} \label{e461x}
   a(0)=0, \quad a(1)=1,\quad a(n)=ra(n-1)-a(n-2)~,\quad n\geq 2~.
\end{eqsp}
Define the sets 
\begin{eqsp}\label{eq:def_S123}
    S_1&:=\left\{\frac{a(n)}{a(n+1)}:n\geq 0\right\}=\left\{0,{1\over r}, {r\over r^2-1}, {r^2-1\over r^3-2r}, \cdots \right\}~,
\\
    S_2&:=\left\{\frac{a(n+1)}{a(n)}:n\geq 0\right\}=\left\{\infty , r, {r^2-1\over r}, {r^3-2r\over r^2-1}, \cdots \right\}~,
    \\
    S_3&:= \begin{cases} \{x: u_c<x<u_c^{-1} \}, \quad {\rm for} \ r\ge 3\cr
    \varnothing, \quad{\rm for}  \ r=1,2\end{cases}\,~ ,
\end{eqsp}
where,
\begin{equation} \label{eucdef}
    u_c=\frac{r-\sqrt{r^2-4}}{2}~,\quad u_c^{-1}=\frac{r+\sqrt{r^2-4}}{2}~.
\end{equation}
Then for $\Omega\in\CR$, if 
\begin{eqsp}\label{eq:heaviside_const_n}
    H(a_nc_n\tau_2 + b_nd_n\sigma_2 + (a_nd_n+b_nc_n)z_2)=1\quad\text{for all}\quad n\in\IZ
\end{eqsp}
then
\ben\label{erequirement}
&&  {c\over a}, \ {d\over b}, \ {c-d\over a-b}\in S_1\cup S_2\cup S_3~.
\een
\end{lemma}
\begin{proof}
We will give a geometric proof of this lemma. For $n=0$, \eqref{eq:heaviside_const_n} implies
\begin{eqnarray}\label{ez2zero}
    z_2'>0~,
\end{eqnarray}
where $\tau_2',\sigma_2',z_2'$ have been defined in \refb{edefprimenew}. 
Let us now define 
\be \label{expypdef}
x' := {z_2'\over \tau_2'}= {ac\tau_2+bd\sigma_2+(ad+bc)z_2\over
a^2\tau_2+b^2\sigma_2+2abz_2}, 
\qquad y' := {\sigma_2'\over \tau_2'}= {c^2\tau_2+d^2\sigma_2+2cdz_2 
\over a^2\tau_2+b^2\sigma_2+2abz_2}\, .
\ee
\refb{ez2zero} can then be rewritten as
\be
x'>0\, .
\ee
For all the Heaviside functions in \eqref{eq:heaviside_const_n} to be nonzero,  
the constraint $x'>0$ on $\tau_2',\sigma_2',z_2'$ must be satisfied if we replace 
$a,b,c,d$ in the expression for $x'$ by $a_n,b_n,c_n,d_n$ respectively. 
We can obtain these restrictions by mapping the line $x'=0$ in the $x'$-$y'$ plane by
the $\mathrm{PSL}(2,\mathbb{Z})$ transformations
\be \label{e555}
g_r :=\begin{pmatrix} 0 & -1\cr 1 & -r\end{pmatrix}, \qquad
g_r^{-1} := \begin{pmatrix} 0 & -1\cr 1 & -r\end{pmatrix}^{-1}~.
\ee
Under the $\mr{PSL}(2,\mathbb{Z})$ transformation
\be
\begin{pmatrix} \tau_2' & z_2'\cr z_2' & \sigma_2' \end{pmatrix}
\longmapsto  \tilde{g} \, \begin{pmatrix} \tau_2' & z_2'\cr z_2' & \sigma_2' \end{pmatrix}\, 
\tilde{g}^t, 
\qquad \tilde{g}:= \begin{pmatrix} \alpha & \beta\cr \gamma & \delta \end{pmatrix}\, ,
\ee
$(x',y')$ maps to
\be \label{e471y}
x ' \longmapsto  {\alpha\gamma + \beta \delta y' + (\alpha\delta + \beta\gamma) x'
\over \alpha^2 + \beta^2 y' + 2\alpha\beta x'}, 
\qquad  y' \longmapsto  {\gamma^2 + \delta^2 y' + 2\gamma \delta x'
\over \alpha^2 + \beta^2 y' + 2\alpha\beta x'}\, .
\ee
Using the condition $y'\ge (x')^2$ for points in $\IH_2$, we can see that the 
denominator of each of the two terms in \refb{e471y} is always positive. 
Therefore for any $\tilde{g}$ of the form $g_r^n$, $n\in\mathbb{Z}$, 
we must have
\be \label{e472x}
\alpha\gamma + \beta \delta y' + (\alpha\delta + \beta\gamma) x' > 0\, .
\ee
These boundaries are
a subset of the straight lines of the form \eqref{epolesimaginary} some of which are shown in Figure~\ref{figfour} except that now
we are drawing them in the $(x',y')$ plane instead of in the $(x,y)$ plane.

To solve the problem of determining the allowed region in the $(x',y')$ plane, we
label the lines obtained by setting the LHS of \refb{e472x} to zero 
by the points where they intersect the parabola 
$y'=(x')^2$. These can be specified by specifying the $x'$ value of the point, which
is the convention we shall follow. 
The only exception is the $x'=$~constant lines whose one end is at $y'=\infty$. We
shall denote this point by $\infty$. In this convention, the original $x'=0$ line,
connecting $\infty$ to 0,
will be represented as $[\infty,0]$ and the allowed region is to the left of this line.
The transformation \refb{e471y} maps this line to
$[\delta/\beta, \gamma/\alpha]$.

It is easy to prove recursively that 
\begin{eqsp} \label{e473x}
    g_r^n&=(-1)^{n+1}\begin{pmatrix}
        a(n-1)&-a(n)\\a(n)&-a(n+1)
    \end{pmatrix}\in\mathrm{PSL}(2,\IZ),\quad n\geq 1~,
    \\
    g_r^{-n}&=(-1)^{n+1}\begin{pmatrix}
        -a(n+1)&a(n)\\-a(n)&a(n-1)
    \end{pmatrix}\in\mathrm{PSL}(2,\IZ),\quad n\geq 1~,
\end{eqsp}
where the $a(n)$'s have been defined in \refb{e461x}.
Under \refb{e473x} the line $(\infty,0)$ maps to 
\be \label{e474x}
\left[{a(n+1)\over a(n)}, {a(n)\over a(n-1)}\right] \quad {\rm and} \quad 
\left[{a(n-1)\over a(n)}, {a(n)\over a(n+1)}\right]\, ,
\ee
respectively. These are lines connecting two successive points in the sets
$S_1$ and $S_2$ respectively. To understand the $n\to\infty$ limits of these
points 
we need
to find the fixed points of the transformation \refb{e471y} under $g_r$. 
Demanding that  $(x', y'=x^{\prime 2})$ maps to 
$(x', y'=x^{\prime 2})$ under the map \refb{e471y} for $\tilde{g}=g_r$, we get
\be
x' =  (r \, x^{\prime 2} - x') / x^{\prime 2}\, .
\ee
This gives a quadratic equation for $x'$ with two solutions
\be 
x' = u_c, u_c^{-1}\, , 
\ee
with $u_c$ given in \refb{eucdef}.

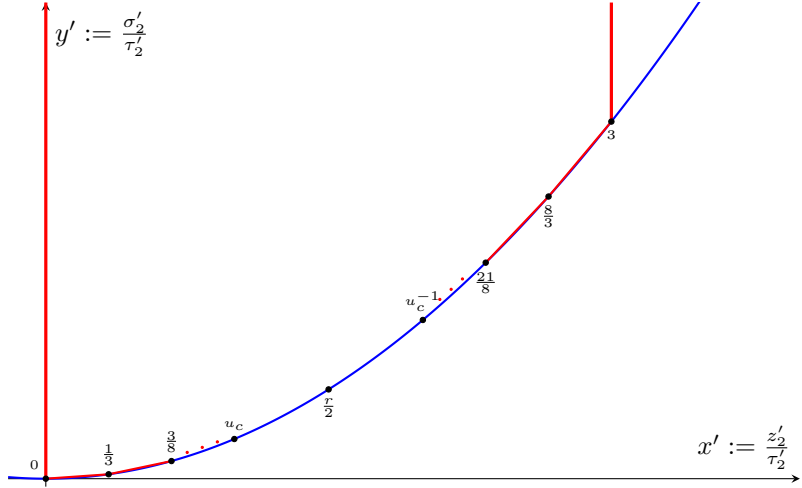
\begin{figure}
    \centering
\begin{tikzpicture}
  \begin{axis}[
      axis lines=middle,
      xlabel={$x':=\frac{z'_2}{\tau'_2}$},
      ylabel={$y':=\frac{\sigma'_2}{\tau'_2}$},
      xmin=-0.2, xmax=4,
      ymin=-0.2, ymax=12,
      domain=-0.2:3.5,
      samples=200,
      width=12cm, height=8cm,
      xtick=\empty,
      ytick=\empty,
      tick label style={font=\small},
    ]

    \addplot[red, very thick] coordinates {(0,0) (0,12)};
    \addplot[red, very thick] coordinates {(3,3^2) (3,12)};

    \addplot[blue, thick] {x^2};

    \addplot[red, thick] coordinates {(0,0) (1/3,{(1/3)^2})};
    \addplot[red, thick] coordinates {(1/3,{(1/3)^2}) (2/3,{(2/3)^2})};
    \addplot[red, thick] coordinates {(8/3,{(8/3)^2}) (3,{3^2})};
    \addplot[red, thick] coordinates {(7/3,{(7/3)^2}) (8/3,{(8/3)^2})};

    \addplot[red, only marks, mark=*, mark size=0.4pt] coordinates {
        (0.75,{0.75^2 + 0.1})
        (0.83,{0.83^2 + 0.1})
        (0.91,{0.91^2 + 0.1})
    };

    \addplot[red, only marks, mark=*, mark size=0.4pt] coordinates {
        (2.09,{2.09^2 + 0.15})
        (2.15,{2.15^2 + 0.15})
        (2.21,{2.21^2 + 0.15})
    };

    \addplot[only marks, mark=*, mark size=1pt] coordinates
      {
        (0,0)
        (1/3,{(1/3)^2})
        (2/3,{(2/3)^2})
        (1,1)
        (1.5,1.5^2)
        (2,4)
        (3,9)
        (7/3,{(7/3)^2})
        (8/3,{(8/3)^2})
      };

    \node[font=\tiny, anchor=south] at (axis cs:-0.06,0) {$0$};
    \node[font=\tiny, anchor=south] at (axis cs:1/3,{(1/3)^2}) {$\tfrac{1}{3}$};
    \node[font=\tiny, anchor=south] at (axis cs:2/3,{(2/3)^2}) {$\tfrac{3}{8}$};
    \node[font=\tiny, anchor=south] at (axis cs:1,1) {$u_c$};
    \node[font=\tiny, anchor=south] at (axis cs:2,4) {$u_c^{-1}$};
    \node[font=\tiny, anchor=north] at (axis cs:3,9) {$3$};
    \node[font=\tiny, anchor=north] at (axis cs:7/3,{(7/3)^2}) {$\tfrac{21}{8}$};
    \node[font=\tiny, anchor=north] at (axis cs:8/3,{(8/3)^2}) {$\tfrac{8}{3}$};
    \node[font=\tiny, anchor=north] at (axis cs:1.5,{(1.5)^2}) {$\frac{r}{2}$};


  \end{axis}
\end{tikzpicture}
    \caption{This figure shows, for $r=3$, the restriction 
    on $(\sigma_2',\tau_2',z_2')$ imposed by the Heaviside functions. 
    The Heaviside functions only allow the region bounded by the red lines. 
    There are infinite number of such red lines shown by the dots.
    The $x'$-axis has been rescaled appropriately to show close points on the 
    curve separately.}
    \label{fig:Heaviside_r=3}
\end{figure}

The lines described in \refb{e474x} 
have been shown in red for $r=3$ in Figure \ref{fig:Heaviside_r=3}. 
The image of the
condition $x'>0$ translates to the condition that $(x',y')$ should lie above the red lines
appearing in Fig.~\ref{fig:Heaviside_r=3}, to the right of the line $x'=0$ and to the 
left of the
line $x'=r$.
This shows that 
part of the parabola $y'=x^{\prime 2}$ 
between $x'=0$ and $x'=u_c$ and between
$x'=1/u_c$ and $x'=r$ are removed by the Heaviside functions except for the isolated
points that appear in the set $S_1\cup S_2$ in \eqref{eq:def_S123}.

Our next task is to translate the condition on $(x',y')$ given above 
to a condition on the
integers $a,b,c,d$ using \refb{edefprimenew}.
A given $(a,b,c,d)$ will map the $\CR$ chamber to one of the
chambers in the $(x',y')$ plane. If this chamber lies in the allowed region
then the corresponding $(a,b,c,d)$ is allowed. Otherwise
it is excluded from the sum in \refb{ef3exp}.
To examine this, we first note that the vertices of the $\CR$ chamber 
shown in Figure~\ref{figfour} are at:
\be
(x,y) := \left(z_2/\tau_2, \sigma_2/\tau_2\right) = (0,0), \ (0,\infty), \ (-1,1)\, .
\ee
Using \refb{expypdef} we can find their
images in the $(x',y')$ plane:
\be \label{e478x}
(x',y') = (c/a, c^2/a^2), \ (d/b, d^2 / b^2), \ ((d-c)/(b-a), (d-c)^2 / (b-a)^2)\, .
\ee
These are all points on the parabola $y'=x^{\prime 2}$. Allowed values of $(a,b,c,d)$
are those for which all the vertices lie in the allowed region in the $(x',y')$ plane.
From Fig.~\ref{fig:Heaviside_r=3} we see that this requires the $x'$ values 
of all the vertices in \refb{e478x} to
either coincide with the set of points in the sets $S_1$ or $S_2$ or lie in the range
$u_c<x<u_c^{-1}$ in which case it is in the set $S_3$. 
This concludes the proof of the lemma.  
\end{proof}

We now decompose the sum into $r=1,2$ and $r\geq 3$ terms:
\begin{eqsp}
    \CF_3=\CF_3^{<}+\CF_3^{>}~,
\end{eqsp}
where 
\begin{eqsp} \label{e460xy}
\CF_3^{<}&=f_{-1}^2 \sum_{r=1,2} r 
\sum_{\big{(}\begin{smallmatrix} a & b\cr c & d\end{smallmatrix}\big{)}\in G_r\backslash \mathrm{PSL}(2,\mathbb{Z})}\hskip .1in 
 \bigg\{\prod_{n=-\infty}^\infty 
H(a_nc_n\tau_2 + b_nd_n\sigma_2 + (a_nd_n+b_nc_n)z_2)
\bigg\}
\\ & \hskip 1in \times \
e^{2\pi i \{(-a^2-c^2+r ac)\tau+(-b^2-d^2+r bd)\sigma+
(-2ab-2cd + r(ad+bc)) z\}}  \,,
\\
\CF_3^{>}&=f_{-1}^2 \sum_{r=3}^{\infty} r 
\sum_{\big{(}\begin{smallmatrix} a & b\cr c & d\end{smallmatrix}\big{)}\in G_r\backslash \mathrm{PSL}(2,\mathbb{Z})}\hskip .1in 
 \bigg\{\prod_{n=-\infty}^\infty 
H(a_nc_n\tau_2 + b_nd_n\sigma_2 + (a_nd_n+b_nc_n)z_2)
\bigg\}
\\ & \hskip 1in \times \
e^{2\pi i \{(-a^2-c^2+r ac)\tau+(-b^2-d^2+r bd)\sigma+
(-2ab-2cd + r(ad+bc)) z\}}  \,.
\end{eqsp}

\begin{prop}\label{prop:F_3>_finite}
The sum $\CF_{3}^{<}$ is a finite sum. 
\begin{proof}
Our aim is to show that for each $r=1,2$, there are finitely many $\mathrm{SL}(2,\IZ)$ matrices which satisfy all the constraints \eqref{eq:heaviside_const_n}. Let us start with $r=1$. We observe that in this case $g_1^3=\mathds{1}$ and $S_1\cup S_2=\{0,1,\infty\}$. Following the proof of Lemma \ref{lemma:Heaviside_matr_const}, we see that if $\big{(}\begin{smallmatrix} a & b\cr c & d\end{smallmatrix}\big{)}\in\mathrm{PSL}(2,\IZ)$ satisfies constraints \eqref{eq:heaviside_const_n}, then 
\begin{eqsp}
{c\over a}, \ {d\over b}, \ {c-d\over a-b}\in\{0,1,\infty\}~.    
\end{eqsp}
The triangle with vertices $(0,1,\infty)$ is a fundamental domain denoted by $\CL$ in Figure \ref{figfour}. There are three $\mathrm{SL}(2,\IZ)$-matrices, 
\be\label{eq:permut_r=12}
\begin{pmatrix} 0 & -1\cr 1 & 0\end{pmatrix}, \qquad 
\begin{pmatrix} 1 & 1\cr 0 & 1\end{pmatrix}, \qquad 
\begin{pmatrix} 1 & 0\cr 1 & 1\end{pmatrix}\, ,
\ee
which map a point $\Omega\in\CR$ to the $\CL$-chamber. 
Indeed, using \refb{eCRmap} 
one can see that the action of these matrices
transform the vertices
$(-1,0,\infty)$ of $\CR$ to $(1,\infty, 0)$, $(\infty, 0,1)$ and 
$(0,1,\infty)$ respectively. 
These matrices are related to each other by multiplication by powers of $g_1$ from the
left.
Thus we conclude that with $r=1$, the sum over $G_r\backslash\mathrm{PSL}(2,\IZ)$ has only one nonvanishing term. Hence the contribution is finite.

\begin{figure}
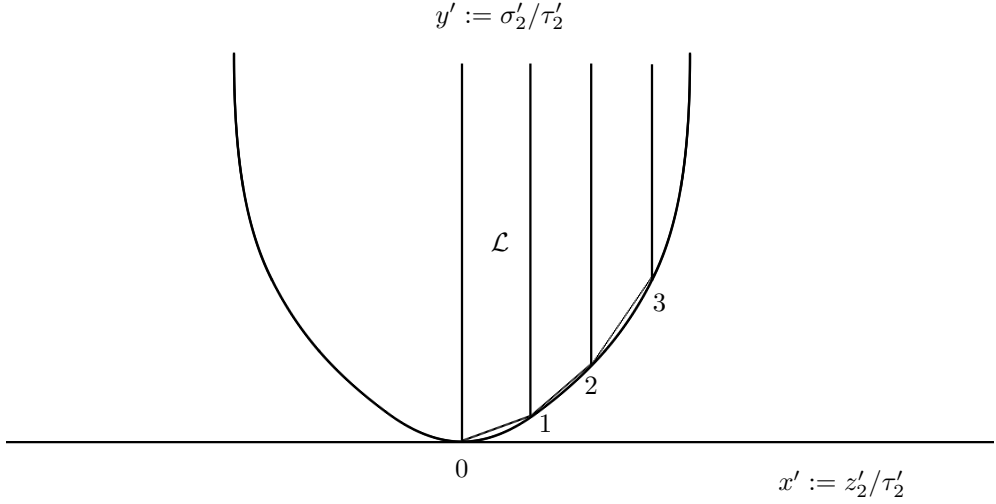

\begin{center}
\figfive
\end{center}

\vskip -.2in

\caption{This figure shows, 
for $r=3$, the three chambers with one vertex at $\infty$. There are other chambers,
not related to these by transformation \refb{e555}, which lie close to the
parabola, with vertices lying between $u_c$ and $u_c^{-1}$.} \label{figfive}
\end{figure}

We now discuss the $r=2$ case. In this case, using Lemma \ref{lemma:Heaviside_matr_const}, 
we see that $u_c=u_{c}^{-1}=1$. Thus $S_3$ is empty. Moreover, Lemma \ref{lemma:Heaviside_matr_const} 
implies that a $\mathrm{PSL}(2,\IZ)$-matrix $\big{(}\begin{smallmatrix} 
a & b\cr c & d\end{smallmatrix}\big{)}$ contributes to the sum only if 
\begin{eqsp}
{c\over a}, \ {d\over b}, \ {c-d\over a-b}\in S_1\cup S_2~.    
\end{eqsp}
In the sum over $G_2\backslash\mathrm{PSL}(2,\IZ)$, we have the freedom of left 
multiplying $\big{(}\begin{smallmatrix} a & b\cr c & d\end{smallmatrix}\big{)}$ by $g_2^n$. \eqref{e471y} and \eqref{e473x} shows that any point in the set $S_1\cap S_2$ can be obtained from $\infty$ by the action of $g_2^{n}$ for some $n\in\IZ$. Thus  
we can assume that $\big{(}\begin{smallmatrix} a & b\cr c & d\end{smallmatrix}\big{)}$ maps $\CR$ to a fundamental domain with one vertex at $\infty$. 
There are only two such fundamental domains in the $0<x'<2$ region, namely the triangles 
with vertices $(0,1,\infty)$ and $(1,2,\infty)$. Furthermore, the domain with vertices $(1,2,\infty)$ 
is obtained from $\CL$ by the action of $g_2=\big{(}\begin{smallmatrix} 0 & -1\cr 1 & -2\end{smallmatrix}\big{)}$. Thus,  we can restrict to the case where the image of $\CR$ under $\big{(}\begin{smallmatrix} a & b\cr c & d\end{smallmatrix}\big{)}$ is $\CL$. As in the $r=1$ case, there are three 
$\mathrm{PSL}(2,\IZ)$-matrices which transform the vertices
$(-1,0,\infty)$ of $\CR$ to $(1,\infty, 0)$, $(\infty, 0,1)$ and 
$(0,1,\infty)$ respectively, but unlike the $r=1$ case, they are not related by $G_2$ transformation. 
Thus there are at most 3 matrices in $G_2\backslash\mathrm{PSL}(2,\IZ)$ which give a 
nonvanishing contribution to the sum and the sum is finite. 
\end{proof}
\end{prop}
For the convergence of $\CF_3^{>}$, the idea is to generalize the analysis of the proof of Proposition \ref{prop:F_3>_finite} to $r\geq 3$ using Lemma \ref{lemma:Heaviside_matr_const}. Let us define 
\begin{eqsp} \label{e573}
    \CS_{12}&:=\left\{\begin{pmatrix}
        a&b\\c&d
    \end{pmatrix}\in\mathrm{PSL}(2,\IZ):\left\{{c\over a},{d\over b},{c-d\over a-b}\right\}\bigcap (S_1\cup S_2)\neq \varnothing\right\}~,
    \\
    \CS_{3}&:=\left\{\begin{pmatrix}
        a&b\\c&d
    \end{pmatrix}\in\mathrm{PSL}(2,\IZ):{c\over a}, {d\over b},{c-d\over a-b}\in S_3\right\}~.
\end{eqsp}
Then we can write 
\begin{eqsp}
    \CF_3^>=\CF_3^{>12}+\CF_3^{>3}~,
\end{eqsp}
where 
\begin{align}
\CF_3^{>12}&:=f_{-1}^2 \sum_{r=3}^{\infty} r 
\sum_{\big{(}\begin{smallmatrix} a & b\cr c & d\end{smallmatrix}\big{)}\in G_r\backslash \CS_{12}}\hskip .1in 
e^{2\pi i \{(-a^2-c^2+r ac)\tau+(-b^2-d^2+r bd)\sigma+
(-2ab-2cd + r(ad+bc)) z\}}  \,,\label{eq:CF3>12_def}
\\
\CF_3^{>3}&:=f_{-1}^2 \sum_{r=3}^{\infty} r 
\sum_{\big{(}\begin{smallmatrix} a & b\cr c & d\end{smallmatrix}\big{)}\in G_r\backslash \CS_{3}}\hskip .1in 
e^{2\pi i \{(-a^2-c^2+r ac)\tau+(-b^2-d^2+r bd)\sigma+
(-2ab-2cd + r(ad+bc)) z\}}  \,.
\end{align}
\begin{prop}
The series $\CF_3^{>12}$ converges absolutely and uniformly on compact subsets of the $\CR$-chamber.     
\end{prop}
\begin{proof}
We first show that, for a fixed $r$, the inner sum is a finite sum with at most $3r$ terms. 
For a fixed $r$, consider a typical term corresponding to $\big{(}\begin{smallmatrix} a & b\cr c & d\end{smallmatrix}\big{)}\in G_r\backslash \CS_{12}$. Let $x',y'$ be as in \eqref{expypdef}. Then 
since one of ${c\over a}, {d\over b},{c-d\over a-b}$ is an element of $\CS_{12}$, we can left 
multiply $\big{(}\begin{smallmatrix} a & b\cr c & d\end{smallmatrix}\big{)}$ 
by $g_r^n$ for some $n$ such that $(x',y')$ belongs to a chamber with one vertex at $\infty$. 
There are exactly
$r$ chambers with one vertex at infinity, the other two vertices being at $x'=\{0,1\}$,
$x'=\{1,2\}$, $\cdots$ $x'=\{r-1,r\}$. 
This has been illustrated in Fig.~\ref{figfive} for the case $r=3$.
Taking into account the cyclic permutation of the three vertices, we get at most 
$3r$ independent matrices $\big{(}\begin{smallmatrix} a & b\cr c & d\end{smallmatrix}\big{)}$
that are not related by left multiplication by the matrices \refb{e555}.  The actual number is less, since for example the chamber with vertices $(r-1,r,\infty)$ is obtained from he one with vertices at $(0,1,\infty)$ by the action of $g_r=\big{(}\begin{smallmatrix} 0 & -1\cr 1 & -r\end{smallmatrix}\big{)}$.

We now show the convergence of the sum over $r$. For this we shall first identify the $3r$ independent matrices described above.
Let $(x',y')$
be the image of $(x,y)$ in the chamber $\CL=(0,1,\infty)$. 
As noted in \eqref{eq:permut_r=12}, there are actually three images due to
the cyclic permutation of the vertices. These are generated from the point $(x,y)$ in the
chamber $\CR$ by the action of the matrices
\be
\begin{pmatrix} 0 & -1\cr 1 & 0\end{pmatrix}, \qquad 
\begin{pmatrix} 1 & 1\cr 0 & 1\end{pmatrix}, \qquad 
\begin{pmatrix} 1 & 0\cr 1 & 1\end{pmatrix}\, ,
\ee
which map the vertices
$(-1,0,\infty)$ of $\CR$ to $(1,\infty, 0)$, $(\infty, 0,1)$ and 
$(0,1,\infty)$ respectively. We denote the images of $(\tau,\sigma,z)$ under these three
maps by 
\begin{eqsp}\label{eq:tsz_primes_def}
   \begin{pmatrix} 0 & -1\cr 1 & 0\end{pmatrix}:  (\tau,\sigma,z)&\longmapsto (\tilde{\tau}',\tilde{\sigma}',\tilde{z}')=(\sigma,\tau,-z)~,
\\
\begin{pmatrix} 1 & 1\cr 0 & 1\end{pmatrix}:(\tau,\sigma,z)&\longmapsto (\tilde{\tau}'',\tilde{\sigma}'',\tilde{z}'')=(\tau+\sigma+2z,\sigma,\sigma+z)~,
\\
\begin{pmatrix} 1 & 0\cr 1 & 1\end{pmatrix}:(\tau,\sigma,z)&\longmapsto (\tilde{\tau}''',\tilde{\sigma}''',\tilde{z}''')=(\tau,\tau+\sigma+2z,\tau+z)~.
\end{eqsp}
The images in the other chambers with a vertex at
$\infty$ can be generated from $\CL$ by the action of the matrices
\be
\begin{pmatrix} 1 & 0\cr n & 1\end{pmatrix}\, , \qquad 1\le n\le r-1\, .
\ee
Indeed, using \refb{esl2zonxy} we see that under this map,
\be
(0,1,\infty) \mapsto (n, n+1,\infty)\, .
\ee
Then we have 
\begin{eqsp}\label{eq:CSr_def}
    G_r\backslash\CS_{12}\subset \left\{\begin{pmatrix} 1 & 0\cr n & 1\end{pmatrix}\begin{pmatrix} 0 & -1\cr 1 & 0\end{pmatrix}, 
\begin{pmatrix} 1 & 0\cr n & 1\end{pmatrix}\begin{pmatrix} 1 & 1\cr 0 & 1\end{pmatrix}, 
\begin{pmatrix} 1 & 0\cr n & 1\end{pmatrix}\begin{pmatrix} 1 & 0\cr 1 & 1\end{pmatrix}:0\leq n\leq r-1\right\}~,
\end{eqsp}
where the containment is because of the fact that there may be further
identification between these matrices by left multiplication by powers of 
\refb{e555}. Using \refb{edefprime}, we see that left multiplication of
$\begin{pmatrix} a & b\cr c & d\end{pmatrix}$ by $\begin{pmatrix} 1 & 0\cr n & 1\end{pmatrix}$
maps 
\be
(\tau_2',\sigma_2',z_2') \to (\tau_2', \sigma'_2+n^2\tau'_2+ 2n z'_2, z'_2+n\tau'_2)\, .
\ee
Using this and \eqref{eq:CSr_def}, we thus have 
\begin{eqsp}\label{e487x}
\left|\CF_3^{>12}\right|&\leq f_{-1}^2 \sum_{r=3}^{\infty} r 
\sum_{n=0}^{r-1}\left[e^{2\pi \{\tilde{\tau}_2'+(\tilde{\sigma}_2'+n^2\tilde{\tau}_2'+2n\tilde{z}_2')-r(\tilde{z}_2'+n\tilde{\tau}_2')\}}+e^{2\pi \{\tilde{\tau}_2''+(\tilde{\sigma}_2''+n^2\tilde{\tau}_2''+2n\tilde{z}_2'')-r(\tilde{z}_2''+n\tilde{\tau}_2'')\}}\right.\\&\left.\hspace{7cm}+e^{2\pi \{\tilde{\tau}_2'''+(\tilde{\sigma}_2'''+n^2\tilde{\tau}_2'''+2n\tilde{z}_2''')-r(\tilde{z}_2'''+n\tilde{\tau}_2''')\}}\right]
\\
&= f_{-1}^2 \sum_{r=3}^{\infty} r 
\sum_{n=0}^{r-1} 
\left[e^{2\pi\{\tilde{\sigma}_2'+\tilde{\tau}_2'-n(r-n)\tilde{\tau}_2'
+ (2n-r)\, \tilde{z}_2'\}}+e^{2\pi\{\tilde{\sigma}_2''+\tilde{\tau}_2''-n(r-n)\tilde{\tau}_2''
+ (2n-r)\, \tilde{z}_2''\}}\right.
\\&\hspace{7cm}+\left.e^{2\pi\{\tilde{\sigma}_2'''+\tilde{\tau}_2'''-n(r-n)\tilde{\tau}_2'''
+ (2n-r)\, \tilde{z}_2'''\}}\right]~.
\end{eqsp}
We now use the following bound for $r\geq 3$:
\begin{eqsp} \label{e496yx}
e^{2\pi\{\tilde{\sigma}_2'+\tilde{\tau}_2'-n(r-n)\tilde{\tau}_2'
+ (2n-r)\, \tilde{z}_2'\}} &= e^{2\pi \{\tilde{\sigma}_2'+\tilde{\tau}_2'- n (r-n) (\tilde{\tau}_2'- \tilde{z}_2')
- ((n+1)(r-n-1)+1) \tilde{z}_2'\}}  
\\
&\leq 
 e^{2\pi(\tilde{\sigma}_2'+\tilde{\tau}_2'- \tilde{C}'\, r)}\ \hbox{for $0\le n\le r-1$}, \quad
\tilde{C}' = \min\{\tilde{\tau}'_2-\tilde{z}_2', \tilde{z}'_2\}\, .
\end{eqsp}
Similar bounds hold if $(\tilde{\tau}'_2,\tilde{\sigma}_2', \tilde{z}'_2)$ is replaced by $(\tilde{\tau}''_2,\tilde{\sigma}_2'', \tilde{z}''_2),(\tilde{\tau}''_2,\tilde{\sigma}_2'', \tilde{z}''_2)$ with constants $\tilde{C}'',\tilde{C}'''$ respectively. Moreover one can easily check that $\tilde{C}',\tilde{C}'',\tilde{C}'''>0$. Thus we get 
\begin{eqsp} \label{e497yx}
\left|\CF_3^{>12}\right|&\leq f_{-1}^2 \sum_{r=3}^{\infty} r 
\sum_{n=0}^{r-1} 
\left[e^{2\pi(\tilde{\sigma}_2'+\tilde{\tau}_2'- \tilde{C}'\, r)}+e^{2\pi(\tilde{\sigma}_2''+\tilde{\tau}_2''- \tilde{C}''\, r)}+e^{2\pi(\tilde{\sigma}_2'''+\tilde{\tau}_2'''- \tilde{C}'''\, r)}\right]
\\&=\sum_{r=3}^{\infty} r^2
\left[e^{2\pi(\tilde{\sigma}_2'+\tilde{\tau}_2'- \tilde{C}'\, r)}+e^{2\pi(\tilde{\sigma}_2''+\tilde{\tau}_2''- \tilde{C}''\, r)}+e^{2\pi(\tilde{\sigma}_2'''+\tilde{\tau}_2'''- \tilde{C}'''\, r)}\right]
\\&<\infty~.    
\end{eqsp}
The absolute and uniform convergence on compact subsets of the $\CR$-chamber follows from the continuity of $\tilde{\sigma}_2'+\tilde{\tau}_2',~  \tilde{C}',~\tilde{\sigma}_2''+\tilde{\tau}_2'',~ \tilde{C}'',\,\tilde{\sigma}_2'''+\tilde{\tau}_2''',$ and $ \tilde{C}'''$ in the $\CR$-chamber.
\end{proof}
We are left to prove the convergence of $\CF_3^{>3}$. For this, we note that it will be futile to
try to prove convergence
before we
pick a representative of the cosets $G_r\backslash \CS_3$ in the sum over
$a,b,c,d$, since the sum over the infinite number of
representatives of the coset will lead to divergence. 
With this in mind, we prove the following lemma.
\begin{lemma}\label{lemma:coset_reps}
We have 
\begin{eqsp}
    G_r\backslash\CS_3= \tilde{\CS}_3:=\left\{\begin{pmatrix}
        a&b\\c&d
    \end{pmatrix}\in\CS_3: {2\over r} \leq  {c\over a} < {r\over 2}\right\}~. 
\end{eqsp}
\end{lemma}
\begin{proof}
It follows from \refb{expypdef} and \refb{e555} that on the parabola $y'=x^{\prime 2}$,
if we denote by $(x_n,y_n=x_n^2)$ the image of $(x',y'=x^{\prime 2})$ under the action of
$g_r^n,n>0$, then we have
\be\label{eurecurx}
x_{n+1}= r - x_n^{-1}, \qquad x_{n-1} = (r-x_n)^{-1} \, .
\ee
This gives,
\be\label{einequalx}
x_{n+1}-u_c = {x_n-u_c\over  x_n u_c} >  x_n-u_c\, ,\qquad
u_c^{-1}-x_{n+1} = {u_c^{-1}-x_n\over x_nu_c^{-1}}<u_c^{-1}-x_n 
\ee
as long as $u_c<x_n<u_c^{-1}$. 
This shows that 
at every step, $x$ is driven away from $u_c$ towards $u_c^{-1}$.
Conversely, the action of $g_r^{-1}$ will drive $x$ towards $u_c$. 
Furthermore, considering the first equation in \refb{e471y} with $y'=x'{}^2$ as a function $f_{g_r}:\IR\to\IR$, we have
\begin{eqsp}
    f_{g_r}(x)=r-\frac{1}{x}~.
\end{eqsp}  
We see that $f_{g_r}$ is a continuous function on $(u_c,u_c^{-1})$. Moreover, 
\begin{eqsp} \label{egderivative}
    f_{g_r}'(x)=\frac{1}{x^2}>0~,\quad x\neq 0~,
\end{eqsp}
showing that $f_{g_r}(x)$ is a monotonically increasing function of $x$. Also,
\begin{eqsp} \label{egspecial}
   f_{g_r}(u_c)=u_c~,\quad  f_{g_r}(2/r)=r/2~,\quad f_{g_r}(u_c^{-1})=u_c^{-1}~.
\end{eqsp}
This shows that the region $u_c<x\le 2/r$ is mapped to the region
$u_c<x\le r/2$ and the region $r/2<x<u_c^{-1}$ is mapped to the region
$2/r<x<u_c^{-1}$ under $g_r$.  Furthermore, using \refb{egderivative} and \refb{egspecial}
it is easy to see that any point in the interval $2/r\leq x<r/2$ is mapped outside this range by the action of $g_r$
and $g_r^{-1}$. 
This shows that any $x$ in the
range $u_c<x<u_c^{-1}$ can be mapped to a unique point in the range 
$2/r\leq x<r/2$ by successive action of $g_r$ or $g_r^{-1}$.
This completes the proof of the lemma.     
\end{proof}
We now want to find a bound for the absolute value of the summand of $\CF_3^{>3}$. Let us write 
\begin{eqsp}
    \left|e^{2\pi i \{(-a^2-c^2+r ac)\tau+(-b^2-d^2+r bd)\sigma+
(-2ab-2cd + r(ad+bc)) z\}}\right|=e^{2\pi B}~,
\end{eqsp}
where 
\begin{eqsp}
B:= \sigma_2'+\tau_2'-r z_2' 
= (a^2 + c^2 - r ac)\tau_2 + (b^2 + d^2 - rbd)\sigma_2 - 
\{  r (ad+bc) - 2ab - 2cd\} z_2\, .
\end{eqsp}
We have the following upper bound for $B$.
\begin{lemma}\label{lemma:bound_B}
For $\big{(}\begin{smallmatrix} a & b\cr c & d\end{smallmatrix}\big{)}\in \tilde{\CS}_3$,
\begin{eqsp} \label{eBbound}
    B\leq \left( {2\over r}-{r\over 2} +1 -{1\over 2} \delta_{r,3}\right)\, C(\Omega)\, \mu(a,b,c,d)<0~,
\end{eqsp}
where $\delta_{r,3}$ is the Kronecker delta defined as 
\begin{eqsp}
    \delta_{r,3}=\begin{cases}
        1&\text{if}~~r=3~,
        \\
        0&\text{if}~~r\neq 3~,
    \end{cases}
\end{eqsp}
and
$C(\Omega)$ and $\mu(a,b,c,d)$ have been defined in \eqref{exx4}.
\end{lemma}
\begin{proof}
We first express $B$ as
\be\label{eBexpx}
B = ac \left( {a\over c} + {c\over a} - r\right)\, (\tau_2 + z_2)
+ bd \left( {b\over d} +{d\over b}
- r\right)\, (\sigma_2+z_2) + z_2 \,(a-b)\, (c-d) \left( r \,  - {a-b\over c-d} 
- {c-d\over a-b}\right)
\, .   
\ee
For $\big{(}\begin{smallmatrix} a & b\cr c & d\end{smallmatrix}\big{)}
\in \tilde{\CS}_3$,
we get
\be\label{e4102xy}
\frac{c}{a}+\frac{a}{c}-r \le {2\over r}+{r\over 2} - r = {2\over r} - {r\over 2}\, ,
\ee
and
\ben\label{e4103xy}
&& \frac{d}{b}+\frac{b}{d}-r 
= \frac{c}{a}+\frac{a}{c}-r + {1\over ab} - {1\over cd} \le {2\over r} - {r\over 2} + 
{1\over ab} - {1\over cd}, 
\nonumber \\
&&  \frac{c-d}{a-b}+\frac{a-b}{c-d}-r 
=  \frac{c}{a}+\frac{a}{c}-r - {1\over a(a-b)} + {1\over c(c-d)} \le {2\over r} - {r\over 2} - 
{1\over a(a-b)} + {1\over c(c-d)} \, . \nonumber \\
\een
Since $c/a$, $d/b$ and $(c-d)/(a-b)$ lie in the range $(u_c, u_c^{-1})$, none of $a,b,c,d,(a-b),(c-d)$
vanish. Furthermore, $a$ and $c$ have the same sign, $b$ and $d$ have the same sign and
$(a-b)$ and $(c-d)$ have the same sign. Thus $ab$ and $cd$ have the same sign and $a(a-b)$ and $c(c-d)$ have the same sign. Thus we have 
\begin{eqsp}\label{eq:bound_ab-cd}
\left| {1\over ab} - {1\over cd}\right|\le 1~, \qquad
\left| {1\over a(a-b)} - {1\over c(c-d)}\right|\le 1~.    
\end{eqsp}
For $r=3$ we can prove stricter bound. Let us first note that the $|1/ab-1/cd|$ attains 
maximum when $ab=\pm 1$ and $cd$ is maximized or 
$cd=\pm 1$ and $ab$ is maximized. But for $r=3$, $ab=1$ along with
\begin{eqsp} \label{ebcinequality}
    2/3<c/a<3/2\implies 2/3<bc<3/2\implies bc=1~.
\end{eqsp}
This then implies that $ad=1+bc=2$, which implies that $cd=2$. 
Thus $|1/ab-1/cd|=1/2$. $ab=-1$ is not allowed since in this case the analog of \refb{ebcinequality} gives
$bc=-1$ which implies $ad=0$. Similar arguments show that 
$cd=1$ is not allowed and $cd=- 1$ implies $ab=- 2$. 
Thus $|1/ab-1/cd|=1/2$.
Similarly, we can show that 
\begin{eqsp}
\left| {1\over a(a-b)} - {1\over c(c-d)}\right|\le \frac{1}{2}~,\quad r=3~.    
\end{eqsp}
Thus we have the bound
\be
\left| {1\over ab} - {1\over cd}\right|\le 1 - {1\over 2} \delta_{r,3}, \qquad
\left| {1\over a(a-b)} - {1\over c(c-d)}\right|\le 1 - {1\over 2} \delta_{r,3},
\ee
Substituting these into \refb{e4103xy}, we get
\begin{eqsp} \label{e108xx}
    \frac{d}{b}+\frac{b}{d}-r\leq\frac{2}{r}-\frac{r}{2}+1-\frac{1}{2}\delta_{r,3}~,\quad \frac{d-c}{b-a}+\frac{b-a}{d-c}-r\leq\frac{2}{r}-\frac{r}{2}+1-\frac{1}{2}\delta_{r,3}~.
\end{eqsp}
Using the fact that $ac$, $bd$,
$(a-b)(c-d)$ are positive for $\big{(}\begin{smallmatrix} a & b\cr c & d\end{smallmatrix}\big{)}
\in \tilde{\CS}_3$ and 
$\tau_2+z_2$, $\sigma_2+z_2$ and $-z_2$ are positive and that
\begin{eqsp}
    \frac{2}{r}-\frac{r}{2}+1-\frac{1}{2}\delta_{r,3}<0~,\quad\text{for}~~r\geq3~,
\end{eqsp}
we get, from \eqref{eBexpx}, \eqref{e4102xy}, \refb{e108xx}
\ben \label{eBless}
B  &\le& \left({2\over r} -{r\over 2} + 1 - {1\over 2} \delta_{r,3} \right) \{ ac (\tau_2+z_2) + bd (\sigma_2+z_2)
+ (a-b) (c-d) (-z_2)\} \nonumber \\
&\le & \left({2\over r} -{r\over 2} + 1 - {1\over 2} \delta_{r,3} \right)  \, C(\Omega) \, \mu(a,b,c,d)\, ,
\een
where we used \refb{exx4}, \eqref{e419xy}. This proves the lemma.
\end{proof}
\begin{prop}
The series $\CF_3^{>3}$ converges absolutely and uniformly on compact subsets of the $\CR$-chamber.    
\end{prop}
\begin{proof}
Using Lemma \ref{lemma:coset_reps} and Lemma \ref{lemma:bound_B}, we have 
\ben
\left|\CF_3^{>3}\right|&\leq& f_{-1}^2 \sum_{r=3}^{\infty} r 
\sum_{\big{(}\begin{smallmatrix} a & b\cr c & d\end{smallmatrix}\big{)}\in G_r\backslash \CS_{3}}\hskip .1in 
e^{2\pi B}     \nonumber \\
&\leq& f_{-1}^2 \sum_{r=3}^{\infty} r 
\sum_{\big{(}\begin{smallmatrix} a & b\cr c & d\end{smallmatrix}\big{)}\in \tilde{\CS}_{3}}\hskip .1in 
\exp\left[2\pi \left({2\over r} -{r\over 2} + 1 - {1\over 2} \delta_{r,3} \right)  \, C(\Omega) \, \mu(a,b,c,d)\right]\, .
\een 
Following the same argument that led to \refb{e423xxy}, we get,
\be
\left|\CF_3^{>3}\right|\leq f_{-1}^2 \sum_{r=3}^{\infty} r 
\sum_{K=1}^\infty 4(2K+1)^2 \, 
\exp\left[2\pi \left({2\over r} -{r\over 2} + 1 - {1\over 2} \delta_{r,3} \right)  \, C(\Omega) \, K\right] <\infty\, .
\ee
The absolute and uniform convergence on compact subsets of $\CR$ follows from the continuity of $C(\Omega)$.
\end{proof}
This proves the convergence of the sum appearing in \refb{eguessfin}.

\section{Properties of $\wt F(\Omega)$} \label{ssection5}

In this section we shall study some properties of $\wt F(\Omega)$ given
in \refb{eguessfin}. In particular, we shall prove the following theorems.

\begin{thm}\label{thm:S_mero_poles}
The function 
\ben \label{eSdef}
S(\Omega) &:=& {1\over 2} \sum_{\big{(}\begin{smallmatrix} a & b\cr c & d\end{smallmatrix}\big{)}\in 
\mr{PSL}(2,\mathbb{Z})}
\left(e^{\pi i \{ac\tau + bd\sigma + (ad+bc)z\}} - e^{-\pi i \{ac\tau + bd\sigma + (ad+bc)z\}}\right)^{-2} \nonumber \\ && 
\hskip 1in 
\times\ f_+(a^2\tau +b^2\sigma +2abz) \ f_+(c^2\tau+d^2\sigma+2cd z) \nonumber \\ 
&+&  
\sum_{\big{(}\begin{smallmatrix} a & b\cr c & d\end{smallmatrix}\big{)}\in  \mr{PSL}(2,\mathbb{Z})}\sum_{r>0} r \sum_{p\ge 0} f_p f_{-1} 
 H(ac\tau_2 + bd\sigma_2 + (ad+bc)z_2) \nonumber \\ && \hskip 1in \times\
H\left(-ac\tau_2 - bd\sigma_2 - (ad+bc)z_2 + ra^2\tau_2 + rb^2\sigma_2+2rab z_2
\right) \,
\nonumber \\ && \hskip 1in \times 
e^{2\pi i \{(pa^2-c^2+r ac)\tau+(pb^2-d^2+r bd)\sigma+
(2pab-2cd + r(ad+bc) )z\}} 
\nonumber \\ 
&+&  f_{-1}^2 \sum_{r>0} r 
\sum_{\big{(}\begin{smallmatrix} a & b\cr c & d\end{smallmatrix}\big{)}\in G_r\backslash \mr{PSL}(2,\mathbb{Z})}\hskip .1in 
 \bigg\{\prod_{n=-\infty}^\infty 
H(a_nc_n\tau_2 + b_nd_n\sigma_2 + (a_nd_n+b_nc_n)z_2)
\bigg\}
\nonumber \\ && \hskip 1in \times \
e^{2\pi i \{(-a^2-c^2+r ac)\tau+(-b^2-d^2+r bd)\sigma+
(-2ab-2cd + r(ad+bc)) z\}}  \, ,
\een
that appear in 
\eqref{eguessfin}, 
is a meromorphic function on $\IH_2$ with double poles at 
\be\label{epoleint_thm}
m_1\tau - n_1\sigma + m_2 + j\, z = 0, \qquad
m_1,n_1,m_2\in \mathbb{Z}, \qquad m_1n_1  +
{j^2\over 4} = {1\over 4}\, .
\ee
    
\end{thm}
It turns out that the double poles of $S(\Omega)$ exactly cancel the double poles \eqref{epoles} of $1/\Phi_{10}$ with $n_2=0$.  
Since $\widetilde{F}(\Omega)=1/\Phi_{10}-S(\Omega)$, we get the following theorem. 
\begin{thm}\label{thm:F_rel_Igusa}
$\wt{F}(\Omega)$ does not have any singularity in the region $\det \rm Im \, \Omega>1/4$.
\end{thm}
Using the relation $\widetilde{F}(\Omega)=1/\Phi_{10}-S(\Omega)$, Theorem \ref{thm:S_mero_poles} and the structure of poles \eqref{epoles} of $1/\Phi_{10}$, we obtain the following theorem.  
\begin{thm}\label{thm:sing_deg_gen}
The generating function $\wt F$ admits a meromorphic continuation to all of 
$\IH_2$ with double poles at 
\begin{eqnarray}
&& n_2 (\tau\sigma - z^2) + m_1\tau - n_1\sigma + m_2 + j\, z = 0, \nonumber \\
&&
m_1,n_1,m_2,n_2\in \mathbb{Z}, \qquad n_2\ge 1, \qquad m_1n_1 + m_2 n_2 +
{j^2\over 4} = {1\over 4}\, .
\end{eqnarray}
\end{thm}
The following theorem establishes the equality of $\wt d^*(T)$ and $d^*(T)$ defined in 
\refb{edstarTint} and hence of $F(\Omega)$ and $\wt F(\Omega)$ for $\det \rm Im \, \Omega>1/4$.
We can then conclude that $\widetilde{F}$ is the analytic continuation of the generating function $F$ 
defined in \eqref{edefFO} even outside the region $\det \rm Im \, \Omega>1/4$. 
\begin{thm}\label{thm:sing_cent_att_cont}
The index $\wt d^*(T)$ for single-centered configurations defined via \eqref{ed_star_tildeF}
is given by a contour integral of $\Phi^{-1}_{10}$ over the attractor contour: for $T=\begin{pmatrix}
        m &\ell/2
        \\
        \ell/2 & n
    \end{pmatrix}$, we have
\begin{equation}\label{edstarT}
\wt d^{*}(T) =\left\{ \begin{split}
&(-1)^{\ell+1} \int_{\mathcal{C}_{m,n,\ell}} d\tau d\sigma dz\, e^{-2 \pi i\left(m\tau+n\sigma+\ell z\right)} \frac{1}{\Phi_{10}(\Omega)}~, \quad \hbox{\emph{for} $m\ge 0$, $n\ge 0$, $4mn-\ell^2\ge 0$~,} \\
& 0~, \qquad \hbox{\rm otherwise}~, \,
    \end{split}\right. 
\end{equation}
where the contour $\mathcal{C}_{m,n,\ell}$ is given as 
\begin{equation}
\begin{split}
\mathcal{C}_{m,n,\ell}:\quad \mr{Im}(\tau)=\frac{2n}{\varepsilon}, \quad \mr{Im}(\sigma)=\frac{2m}{\varepsilon}, \quad \mr{Im}(z)=-\frac{\ell}{\varepsilon},\\
0 \leq \mr{Re}(\tau),\mr{Re}(\sigma), \mr{Re}(z)<1~,
\end{split}
\label{eq:att_cont}
\end{equation}
where $\varepsilon>0$ is a small positive number.     
\end{thm}

\subsection{Proof of Theorem \ref{thm:S_mero_poles}}
As described below \eqref{esl3zoffi}, using the manifest invariance of $S(\Omega)$ under $\mr{SL}(2,\IZ)$, the convergence of $S(\Omega)$ on $\IH_2$ except on the hypersurfaces 
\be\label{epoleint_sec}
m_1\tau_2 - n_1\sigma_2 + j\, z_2 = 0, \qquad
m_1,n_1,j\in \mathbb{Z}, \qquad m_1n_1  +
{j^2\over 4} = {1\over 4}\, ,
\ee
follows from the convergence in the $\CR$-chamber. We now show that $S(\Omega)$ is an analytic
function except for double poles on the hypersurface \eqref{epoleint_thm}. We will need the following lemma.
\begin{lemma}\label{lemma:sl2z_linear_poles}
Every hypersurface of the form \eqref{epoleint_sec} is related to the hypersurface $z=0$ by the $\mr{SL}(2,\IZ)$-action $\Omega\to\gamma\Omega\gamma^t$ of some matrix $\gamma\in\mr{SL}(2,\IZ)$ and translation $\Omega\to\Omega+A$ with $A$ a symmetric $2\times 2$ integer-valued matrix.    
\end{lemma}

\begin{proof}
The proof proceeds along the line of the analysis following \refb{edefprime}.
Under the
$\mr{PSL}(2,\mathbb{Z})$ transformation 
\begin{equation}
    \tau':=a^2\tau+b^2\sigma+2abz,\quad \sigma':=c^2\tau+d^2\sigma+2cdz,\quad z':=ac\tau+bd\sigma+(ad+bc)z~,
\end{equation}
we get
\be
- m_1' \tau'+n_1'\, \sigma' + j'\, z' = -m_1 \tau+ n_1 \sigma + j z  \, ,
\ee
with
\be
m_1 =m_1'a^2 - n_1' c^2  + j'ac , \qquad n_1 = - m_1' b^2 + n_1'd^2 - j' bd,
\qquad j = 2 m_1' ab - 2 n_1' cd + (ad+bc) j'\, .
\ee
Furthermore, we have
\begin{equation}
m_1 n_1+{j^2\over 4}={1\over 4} \quad \Leftrightarrow \quad m_1'n_1'+\frac{j'^2}{4}=
{1\over 4}~.    
\end{equation}
This shows that all the  hypersurfaces \refb{epoleint_thm} for $m_2=0$ are
mapped to each other by $\mr{PSL}(2,\IZ)$ transformation. Moreover, the hypersurface \refb{epoleint_thm} with 
\begin{eqsp}
    (m_1,n_1,j,m_2)=(ac,-bd,ad+bc,0)~,\quad \begin{pmatrix}
        a&b\\c&d
    \end{pmatrix}\in\mr{SL}(2,\IZ)~,
\end{eqsp}
is mapped to the hypersurface $z'=0$. Thus, varying over all $\mr{SL}(2,\IZ)$ matrices, we can map an arbitrary  hypersurface \refb{epoleint_thm} for $m_2=0$ to the $z=0$ hypersurface. 
To get $m_2\neq 0$, we can translate $\Omega$ by an arbitrary integer valued 
symmetric matrix $A$. 
\end{proof}

Since according to \refb{esl3zoffi}, $S(\Omega)$ is invariant under $\mathrm{SL}(2,\IZ)$-action 
$\Omega\to\gamma\Omega\gamma^t$
and translations $\Omega\to\Omega+A$ for $A$ an integer-valued symmetric $2\times 2$ matrix,
by Lemma \ref{lemma:sl2z_linear_poles}, it suffices to prove that $S(\Omega)$ has a 
double pole at $z=0$ and then use the $\mr{SL}(2,\IZ)$-invariance and translation 
invariance of $S(\Omega)$. To this end, we notice that as we approach the $z_2=0$ line from the $z_2<0$ or the $z_2>0$ side, the terms corresponding to 
$\big{(}\begin{smallmatrix}
    1&0\\0&1
\end{smallmatrix}\big{)},\big{(}\begin{smallmatrix}
    0&-1\\1&0
\end{smallmatrix}\big{)}$ in the first sum, $\big{(}\begin{smallmatrix}
    0&-1\\1&0
\end{smallmatrix}\big{)},\big{(}\begin{smallmatrix}
    1&0\\r&1
\end{smallmatrix}\big{)}$,
$\big{(}\begin{smallmatrix}
    1&0\\0&1
\end{smallmatrix}\big{)},\big{(}\begin{smallmatrix}
    0&1\\-1&r
\end{smallmatrix}\big{)}$ in the second sum and $\big{(}\begin{smallmatrix}
    0&-1\\1&0
\end{smallmatrix}\big{)}$, 
$\big{(}\begin{smallmatrix}
    1&0\\0&1
\end{smallmatrix}\big{)}$ in the third sum in \eqref{eSdef}
gives  
\begin{eqsp}\label{eq:S_pole_def}
    S_{\mr{pole}}(\Omega)&:=\left(e^{i\pi z}-e^{-i\pi z}\right)^{-2}f_+(\tau)f_+(\sigma) 
    \\ &
    +f_{-1}e^{-2\pi i\tau}\sum_{p\geq 0}f_pe^{2\pi ip\sigma}\sum_{r>0}r 
    \left\{e^{-2\pi irz} H(-z_2)+ e^{2\pi irz} H(z_2)\right\}
    \\
    &+f_{-1}e^{-2\pi i\sigma}\sum_{p\geq 0}f_pe^{2\pi ip\tau}\sum_{r>0}r 
    \left\{e^{-2\pi irz} H(-z_2)+ e^{2\pi irz} H(z_2)\right\} \\ &
    + f_{-1}^2 \, e^{-2\pi i\tau} \, e^{-2\pi i\sigma} \, 
    \sum_{r>0}r \, \left\{e^{-2\pi irz} H(-z_2)+ e^{2\pi irz} H(z_2)\right\}
    \\&=(e^{\pi i z}- e^{-\pi i z})^{-2} f(\sigma) f(\tau)\, .
\end{eqsp}
Note that $S_{\mr{pole}}(\Omega)$ has a double pole at $z=0$ that exactly cancels
the double pole of $1/\Phi_{10}$ at $z=0$ given in \refb{epolestructure}. 
Then we can write\footnote{Despite the use of subscript ``hol'', $\CF_{i,\mr{hol}}$ continue to have poles on the subspaces \eqref{epoleint_thm} except at $z=0$.}
\be \label{espoleresults}
S(\Omega) = S_{\rm pole}(\Omega) + \CF_{1,\rm hol}(\Omega) +  \CF_{2,\rm hol}(\Omega) +  \CF_{3,\rm hol}(\Omega)\, ,
\ee
where
\begin{eqsp}\label{eq:S_hol}
\CF_{1,\rm hol}(\Omega)&:=    {1\over 2} \sum_{\big{(}\begin{smallmatrix} 1 & 0\cr 0 & 1\end{smallmatrix}\big{)},\big{(}\begin{smallmatrix} 0 & -1\cr 1 & 0\end{smallmatrix}\big{)}\neq \big{(}\begin{smallmatrix} a & b\cr c & d\end{smallmatrix}\big{)}\in \mr{PSL}(2,\mathbb{Z})}\hskip .1in 
\left(e^{\pi i \{ac\tau + bd\sigma + (ad+bc)z\}} - e^{-\pi i \{ac\tau + bd\sigma + (ad+bc)z\}}\right)^{-2}  \\ & 
\hskip 1in 
\times\ f_+(a^2\tau +b^2\sigma +2abz) \ f_+(c^2\tau+d^2\sigma+2cd z)~,  \\ 
\CF_{2,\rm hol}(\Omega)&:= \sum_{p\ge 0} f_p f_{-1} \sum_{r>0} r 
\sum_{\big{(}\begin{smallmatrix}
    0&-1\\1&0
\end{smallmatrix}\big{)},\big{(}\begin{smallmatrix}
    1&0\\r&1
\end{smallmatrix}\big{)},
\big{(}\begin{smallmatrix}
    1&0\\0&1
\end{smallmatrix}\big{)},\big{(}\begin{smallmatrix}
    0&1\\-1&r
\end{smallmatrix}\big{)}\neq \big{(}\begin{smallmatrix} a & b\cr c & d\end{smallmatrix}\big{)}\in \mr{PSL}(2,\mathbb{Z})} 
 H(ac\tau_2 + bd\sigma_2 + (ad+bc)z_2)  \\ & \hskip 1in \times\
H\left(-ac\tau_2 - bd\sigma_2 - (ad+bc)z_2 + ra^2\tau_2 + rb^2\sigma_2+2rab z_2
\right) \,
\\ & \hskip 1in \times 
e^{2\pi i \{(pa^2-c^2+r ac)\tau+(pb^2-d^2+r bd)\sigma+
(2pab-2cd + r(ad+bc) )z\}}~, 
 \\ 
\CF_{3,\rm hol}(\Omega)&:=  f_{-1}^2 \sum_{r>0} r 
\sum_{\big{(}\begin{smallmatrix} 0 & -1\cr 1 & 0\end{smallmatrix}\big{)},
\big{(}\begin{smallmatrix} 1 & 0\cr 0 & 1\end{smallmatrix}\big{)}\neq \big{(}\begin{smallmatrix} a & b\cr c & d\end{smallmatrix}\big{)}\in G_r\backslash \mr{PSL}(2,\mathbb{Z})} 
\prod_{n=-\infty}^\infty  
H(a_nc_n\tau_2 + b_nd_n\sigma_2 + (a_nd_n+b_nc_n)z_2)
\\ & \hskip 1in \times \
e^{2\pi i \{(-a^2-c^2+r ac)\tau+(-b^2-d^2+r bd)\sigma+
(-2ab-2cd + r(ad+bc)) z\}}  \, .
\end{eqsp}
We shall first show that $\CF_{1,\rm hol}$,   $\CF_{2,\rm hol}$ and  
$\CF_{3,\rm hol}$ converge in the limit  $z_2\to 0^-$.
\begin{prop}\label{prop:S_hol_conv}
The series $\CF_{1,\rm hol}$,   $\CF_{2,\rm hol}$ and  
$\CF_{3,\rm hol}$ converge absolutely in the limit $z_2\to 0^-$.    
\end{prop}

\begin{proof}
Our strategy will be to reexamine the proof of convergence of 
$\CF_1$, $\CF_2$, $\CF_3$ in $\CR$ given in
Section \ref{stildeFconverge} and identify the steps where we 
used the strict inequality $z_2<0$.
Then we shall show that the inequality $z_2<0$ can be replaced by $z_2\le 0$
once we remove the terms described in \refb{eq:S_hol}.

Reexamining the analysis in Section \ref{stildeFconverge} we see that 
the main step where we need the $z_2<0$ inequality is to ensure that $C(\Omega)$
defined in \refb{exx4} is strictly positive. 
For $z_2=0$, $C(\Omega)=0$ but we can define
a new quantity
\be
\wt C(\Omega) := {\rm Min} \{\tau_2, \sigma_2\}\, ,
\ee 
which  is strictly positive. Therefore if in the
analysis of $\CF_{1,\rm hol}$,   $\CF_{2,\rm hol}$ and  
$\CF_{3,\rm hol}$, we can
replace $C(\Omega)$ by $\wt C(\Omega)$, we would
prove convergence of $S_{\mr{hol}}(\Omega)$ for $\Omega\in \CR\cup \{z_2=0\}$.

The first place in the analysis of Section \ref{stildeFconverge} where 
$C(\Omega)$  was used, is in \refb{ecrucialinequality} in the analysis of
$\CF_1$.
For $z_2=0$, the left hand side of \refb{ecrucialinequality} can be written as
\be\label{etrial}
|ac \tau_2+ bd \sigma_2|\, .
\ee
Recalling that $\mu(a,b,c,d)={\rm Max}(|a|,|b|,|c|,|d|)$, we can see that as long as
$a,b,c,d$ are all non-zero, \refb{etrial} is larger than $\wt C(\Omega)\, \mu(a,b,c,d)$.
Even when one of $a,b,c,d$ vanish, we still have $|ac \tau_2+ bd \sigma_2|
= \wt C(\Omega)\, \mu(a,b,c,d)$. For example if $a=0$ then $bc=-1$ and we can
take $\mu(a,b,c,d)=|d|$. This will give $|ac \tau_2+ bd \sigma_2|
=|d\sigma_2| \geq \wt C(\Omega)\, \mu(a,b,c,d)$. Similar analysis can be done of the cases when
$b$, $c$ or $d$ vanish. The only case where 
$|ac \tau_2+ bd \sigma_2|< \wt C(\Omega)\, \mu(a,b,c,d)$ occurs when two of the
coefficients $a,b,c,d$ vanish. There are two cases:
$\begin{pmatrix} a & b\cr c & d\end{pmatrix}=
\begin{pmatrix} 1 & 0\cr 0 & 1\end{pmatrix}$ or
$\begin{pmatrix} 0 & -1\cr 1 & 0\end{pmatrix}$. In both cases 
$|ac \tau_2+ bd \sigma_2|=0$ but $ \wt C(\Omega)\, \mu(a,b,c,d)=\mr{Min}\{\tau_2,\sigma_2\}$. However
in the expression for $\CF_{1,\rm hol}$ given in \refb{eq:S_hol},
these terms have been removed. 
Therefore we conclude that $\CF_{1,\rm hol}$ is convergent
in $\CR\cup\{z_2=0\}$.

Next we analyze the contribution to $\CF_{2,\rm hol}$.
The first place where the $C(\Omega)>0$ condition was used 
was in \refb{ecomone} and \refb{ecomtwo} during the
analysis of $\CF^>_{2}$.
Of these, in \refb{ecomtwo} we can replace $C(\Omega)$ by $\wt C(\Omega)$ since
$a$ and $b$ cannot vanish simultaneously, but in \refb{ecomone} we need as
usual the conditions $\begin{pmatrix} a & b\cr c & d\end{pmatrix}\ne$ 
$\begin{pmatrix} 1 & 0\cr 0 & 1\end{pmatrix}$ or
$\begin{pmatrix} 0 & -1\cr 1 & 0\end{pmatrix}$. These terms have been 
removed from the definition of
$\CF_{2}$. Hence we see that we can replace $C(\Omega)$ by $\wt C(\Omega)$
in this analysis.

$C(\Omega)$ also appears in \refb{e454xxz} in the proof of convergence of
$\CF^{<\epsilon}_{2}$. The argument for replacing $C(\Omega)$ by $\wt C(\Omega)$
is identical to the one given above.

The third place where $C(\Omega)$ appears in the analysis of $\CF_2$ is in
analyzing the convergence of \refb{eq:ser_>eps} that appears in the 
expression for $\CF^{>\epsilon}_{2}$.  In this case the matrix
$\begin{pmatrix} a & b\cr c & d\end{pmatrix}$ is replaced by 
$\begin{pmatrix} a & b\cr -k & -\ell\end{pmatrix}$ and as before we have
two possible choices for this matrix where we cannot replace $C(\Omega)$ by
$\wt C(\Omega)$. 
To translate this into a condition on $a,b,c,d$, we use
\be
(c,d)=(a\,n-k, b\,n-\ell) = (a\, \r0-k, b\, \r0-\ell) = (ar-k, br-\ell)~,
\ee
where we used the fact that $n:=\r0$ and the fact that 
$\CF^{>\epsilon}_{2}$ is part of the contribution to $\CF_2$ where we
have set $r=\r0$. Setting $\begin{pmatrix} a & b\cr -k & -\ell\end{pmatrix}=\begin{pmatrix} 1 & 0\cr 0 & 1\end{pmatrix}$ and
$\begin{pmatrix} 0 & 1\cr -1 & 0\end{pmatrix}$ gives the matrices $\begin{pmatrix} a & b\cr c & d\end{pmatrix}
=\begin{pmatrix} 1 & 0\cr r & 1\end{pmatrix}, \begin{pmatrix} 0 & 1\cr -1 & r\end{pmatrix}$,
both of which  are removed in the definition of
$\CF_{2,\rm hol}$ in \refb{eq:S_hol}. This shows that $\CF_{2,\rm hol}$
represents a convergent sum for $\Omega\in \CR\cup\{z_2=0\}$.

In the analysis of $\CF_3$ we used $C(\Omega)$ in \refb{eBless}. However in that
case we had $\big{(}\begin{smallmatrix} a & b\cr c & d\end{smallmatrix}\big{)}
\in \tilde{\CS}_3$. Since the set $ \tilde{\CS}_3$ does not include 
either $\begin{pmatrix} 1 & 0\cr 0 & 1\end{pmatrix}$ or 
$\begin{pmatrix} 0 & -1\cr 1 & 0\end{pmatrix}$, we can replace $C(\Omega)$ by
$\wt C(\Omega)$. 

The $z_2<0$ condition was used in another place in the analysis of $\CF_3$
that does not directly use $C(\Omega)$. In the analysis of \refb{e496yx},
\refb{e497yx} we needed the conditions
\be
\tilde{C}'={\rm Min}\{\tilde{\tau}'_2-\tilde{z}_2', \tilde{z}'_2\}>0, \qquad
\tilde{C}''={\rm Min}\{\tilde{\tau}''_2-\tilde{z}_2'', \tilde{z}''_2\}>0, \qquad
\tilde{C}''={\rm Min}\{\tilde{\tau}'''_2-\tilde{z}_2''', \tilde{z}'''_2\}>0,
\ee
where $\tilde{\tau}'_2$, $\tilde{z}'_2$, $\tilde{\tau}''_2$, $\tilde{z}''_2$, 
$\tilde{\tau}'''_2$, $\tilde{z}'''_2$ have been defined in
\refb{eq:tsz_primes_def}. For $z_2=0$, using \eqref{eq:tsz_primes_def}, we see that 
\begin{eqsp}
    \tilde{C}'=0~,\quad \tilde{C}''=\tau_2>0~,\quad \tilde{C}'''=0~.
\end{eqsp}
Thus, using the bound in \eqref{e496yx}, the second sum in the third line of \eqref{e487x} converges for $z_2=0$ but to prove the convergence of first sum in third line of \eqref{e487x} and fourth line of \eqref{e487x}, we need to examine the exponents more closely. Using \eqref{eq:tsz_primes_def}, the exponent in the first term in the third line of \eqref{e487x} for $z_2=0$ can be written as
\begin{eqsp}
    2\pi (\tau_2+\sigma_2-n(r-n)\sigma_2)~.
\end{eqsp}
Thus, only the $n=0$ term causes a problem in the convergence of this sum.  Similarly, the exponent of the fourth line of \eqref{e487x} for $z_2=0$ is given by
\begin{eqsp}
    2\pi (2\tau_2+\sigma_2-n(r-n)\tau_2+(2n-r)\tau_2)=2\pi (\tau_2+\sigma_2-(n+1)(r-1-n)\tau_2)~.
\end{eqsp}
Thus, only the $n=r-1$ term causes a problem in the convergence of this sum. 
From \eqref{eq:CSr_def}, we see that these two problematic cases correspond to the matrices
\begin{eqsp}
\begin{pmatrix} a & b\cr c & d\end{pmatrix}=\begin{pmatrix} 0 & -1\cr 1 & 0\end{pmatrix}~,\quad \begin{pmatrix} a & b\cr c & d\end{pmatrix}=\begin{pmatrix} 1 & 0\cr r-1 & 1\end{pmatrix}\begin{pmatrix} 1 & 0\cr 1 & 1\end{pmatrix}=\begin{pmatrix} 1 & 0\cr r & 1\end{pmatrix}=g_r\begin{pmatrix} 0 & 1\cr -1 & 0\end{pmatrix}~.    
\end{eqsp}
Since this is excluded from
the definition of $\CF_{3,\rm hol}$ in \refb{eq:S_hol}, we conclude that
$\CF_{3,\rm hol}$  represents a convergent sum. 
\end{proof}

This shows that $S(\Omega) - S_{\rm pole}(\Omega)$ has a convergent expansion as $z_2\to 0^-$
from the $\CR$ chamber. 
S-duality transformation by $\gamma_0:=\begin{pmatrix} 0 & -1\cr 1 & 0\end{pmatrix}$
exchanges $\tau$ and $\sigma$ and maps $z$ to $-z$. This
takes the $\CR$ chamber to the $\CL$ chamber  leaving fixed the boundary $z_2=0$. This
also leaves $S_{\rm pole}$ invariant. 
Furthermore, since the set of
matrices $\begin{pmatrix} a & b\cr c & d\end{pmatrix}$ that have been
removed from the sum in \refb{eq:S_hol} is invariant under right multiplication
by $\gamma_0$,
this also leaves $\CF_{1,\rm hol}$, $\CF_{2,\rm hol}$ and $\CF_{3,\rm hol}$ unchanged.
This shows that $S(\Omega) - S_{\rm pole}(\Omega)$ also has a convergent expansion as
$z_2\to 0^+$ from the $\CL$ chamber. Finally the removal of the special matrices from the
sum in \refb{eq:S_hol} has removed all factor of $H(\pm z_2)$ from the sum and hence
$\CF_{1,\rm hol}$, $\CF_{2,\rm hol}$ and $\CF_{3,\rm hol}$ 
is analytic in $\CR\cup\CL\cup\{z_2=0\}$. 

This proves the holomorphicity of $S(\Omega) - S_{\rm pole}(\Omega)$ in the region $\CR\cup\CL\cup\{z_2=0\}$.
$\mathrm{PSL}(2,\IZ)$-invariance of
$S(\Omega)$ then establishes
Theorem \ref{thm:S_mero_poles}.

\subsection{Proof of Theorem \ref{thm:F_rel_Igusa}}\label{sec:hol_Ftilde}
\begin{proof}
We need to prove that the function $\wt{F}(\Omega)$ given by the RHS of \eqref{eguessfin} is holomorphic in the domain $\det\rm Im \, \Omega)>1/4$. First, note that, as proved below \eqref{eq:poles_new},
$1/{\Phi_{10}}$ has no poles of the form \eqref{epoles} with $n_2\neq 0$ in the domain $\det \rm Im \, \Omega>1/4$ \cite{Sen:2007qy}. 
Eqs.\refb{eq:S_pole_def} and \refb{espoleresults} shows us that $S(\Omega)$ 
is an analytic function in $\CR\cup\CL\cup\{z_2=0\}$ except for a double pole at $z=0$
where it behaves as
\be
(e^{\pi i z}- e^{-\pi i z})^{-2} f(\sigma) f(\tau)\, .
\ee
This is precisely how $1/\Phi_{10}$ behahves at $z=0$. 
Therefore $1/\Phi_{10}(\Omega) - S(\Omega)$ is analytic in the intersection of the domains
$\CR\cup\CL\cup\{z_2=0\}$ and $\mr{det\,Im\,}\Omega>\frac{1}{4}$.
$\mathrm{PSL}(2,\IZ)$-invariance of
$S(\Omega)$ and $1/\Phi_{10}(\Omega)$ then establishes that
$\wt F(\Omega) =1/\Phi_{10}(\Omega) - S(\Omega)$ does not have a pole
in the domain $\det\rm Im \, \Omega)>1/4$, since any point in this domain can be
mapped to $\CR\cup\{z_2=0\}$ under an appropriate $\mathrm{PSL}(2,\IZ)$-transformation.
\end{proof}

\subsection{Proof of Theorem \ref{thm:sing_deg_gen}}\label{sec:Ftilde_ana_cont}
\begin{proof}
By Theorem \ref{thm:S_mero_poles}, the RHS of \eqref{eguessfin} defines the meromorphic continuation of $\widetilde{F}$ to the entire $\IH_2$. From the proof of Theorem \ref{thm:F_rel_Igusa}, we see that all the poles of $1/\Phi_{10}$ of the form \eqref{epoles} with $n_2=0$ is canceled by $S(\Omega)$. Thus the only poles of $\widetilde{F}$ are of the form \eqref{epoles} with $n_2\neq 0$. Using the invariance of the hypersurface \eqref{epoles} under the transformation  
\begin{eqsp}
    (m_1,n_1,m_2,n_2,j)\to (-m_1,-n_1,-m_2,-n_2,-j)~,
\end{eqsp}
we can restrict to $n_2\geq 1$. 
\end{proof}

\subsection{Proof of Theorem \ref{thm:sing_cent_att_cont}} \label{sconvergenceproof}

\begin{proof}
We shall now give a proof of formula \refb{edstarT}. Due to Theorem \ref{thm:F_rel_Igusa}, $\widetilde{F}(\Omega)$ does not have any singularity in the domain $\mr{det\,Im}\,\Omega>\frac{1}{4}$. 
Hence its Fourier expansion is unique in that domain:
\be\label{ed_star_tildeF}
\wt d^*(m,n,\ell) = (-1)^{\ell+1} \int_{\widetilde{\mathcal{C}}} d\tau d\sigma dz\, e^{-2 \pi i\left(m\tau+n\sigma+\ell z\right)} \wt{F}(\Omega)\, ,
\ee
for any contour $\wt{\mathcal{C}}\subset \det \rm Im \, \Omega > {1\over 4}$ of the form
\be \label{eq 3.3new}
\wt{\mathcal{C}}: 0\leq \tau_1,\sigma_1,z_1< 1\,,\quad \tau_2,\sigma_2,z_2 ~~\text{fixed}~.
\ee
In view of \refb{ed_star_tildeF}, 
to prove the desired result, we need to show that
\be\label{eequality}
\int_{\widetilde{\mathcal{C}}}d\tau d\sigma dz\, e^{-2 \pi i\left(m\tau+n\sigma+\ell z\right)} \wt{F}(\Omega) = \int_{\mathcal{C}_{m,n,\ell}} d\tau d\sigma dz\, e^{-2 \pi i\left(m\tau+n\sigma+\ell z\right)} \frac{1}{\Phi_{10}(\Omega)}~,
\ee
for $\CC_{m,n,\ell}$ given by \eqref{eq:att_cont} and for any choice of contour $\widetilde{\mathcal{C}}$ of the form given in
\refb{eq 3.3new}.

We first consider charges $(m,n,\ell)$ with
\be\label{econdition}
m\geq 0, \qquad n\ge 0, \qquad 4mn-\ell^2\ge 0\, . 
\ee
It follows from this that the charge matrix $T$ has non-negative eigenvalues:
\be
T\ge \mathbb{0}\, .
\ee
Since by Theorem \ref{thm:F_rel_Igusa}, $\wt F(\Omega)$ is analytic
for $\det\,$Im$\,\Omega>1/4$, it is sufficient to prove \refb{eequality} for the choice
$\widetilde{\mathcal{C}} =\mathcal{C}_{m,n,\ell}$. We shall proceed with this choice.
Since $\wt F(\Omega)=1/\Phi_{10}(\Omega) - S(\Omega)$, we need to show that
\be \label{esomegaidentity}
\int_{\mathcal{C}_{m,n,\ell}} d\tau d\sigma dz\, e^{-2 \pi i\left(m\tau+n\sigma+\ell z\right)}
S(\Omega)\, ,
\ee
vanishes.

We shall first show that
the terms proportional
to $f_pf_{-1}$ and $f_{-1}^2$ in \refb{eSdef} 
do not contribute to \refb{esomegaidentity}.
First take $(a,b,c,d)=(1,0,0,1)$ in \refb{eSdef}.
In this case the coefficients of the $f_pf_{-1}$ and $f_{-1}^2$ terms represent contributions
from the charge matrix with eigenvalues $(-1,p)$ and $(-1,-1)$ respectively. Since 
$\gamma^t T \gamma$ for $\gamma\in \mr{SL}(2,\mathbb{Z})$ will have eigenvalues whose signs
match those of $T$, we see that each term in \refb{eSdef} multiplying $f_{-1}f_p$
or $f_{-1}^2$ will have at least one eigenvalue of $T$ negative. Hence they will not contribute to
\refb{eequality} for $T\ge \mathbb{0}$.

Thus we are left with the
terms in the first two lines of \refb{eSdef}. 
We need to show that these also do not contribute to
\refb{esomegaidentity}. 
We shall first prove this for the choice $(a,b,c,d)=(1,0,0,1)$; and then extend the proof to the other terms using $\mr{SL}(2,\IZ)$-invariance. The relevant term is:
\begin{equation}\label{eendT}
\int_{\mathcal{C}_{m,n,\ell}} d\tau d\sigma dz\, e^{-2 \pi i\left(m\tau+n\sigma+\ell z\right)} 
(e^{\pi i z} - e^{-\pi i z})^{-2}\,
f(\sigma) f(\tau)\, .
\end{equation}
First consider the case when $\ell>0$. In that case it follows from \eqref{eq:att_cont} that 
on $\mathcal{C}_{m,n,\ell}$, 
${\rm Im}(z)<0$. In this region, the $(e^{\pi i z} - e^{-\pi i z})^{-2}$ factor
in \refb{eendT} can be expanded to 
\be 
(e^{\pi i z} - e^{-\pi i z})^{-2} = \sum_{k=1}^\infty k\, e^{- 2\pi i k z}. 
\ee
Substituting this into \refb{eendT} we see that the integral over $z$ vanishes since $k,\ell>0$.
A similar argument holds for the case $\ell<0$. For $\ell=0$ the integrand has a double
pole at $z=0$ but does not have a single pole term, and the integral vanishes
for either choice of the contour. 

For general $(a,b,c,d)$, we express the integrand as 
\begin{eqsp}
    e^{-2\pi i(m'\tau'+n'\sigma'+\ell'z')}(e^{\pi iz'}-e^{-i\pi z'})^{-2}~,
\end{eqsp}
where 
\begin{eqsp}
    \begin{pmatrix}
        m'&\ell'/2\\\ell'/2&n'
    \end{pmatrix}:=(\gamma^{-1})^t\begin{pmatrix}
        m&\ell/2\\\ell/2&n
    \end{pmatrix}\gamma^{-1}~,\quad \begin{pmatrix}
        \tau'&z'\\z'&\sigma'
    \end{pmatrix}:=\gamma\begin{pmatrix}
        \tau&z\\z&\sigma
    \end{pmatrix}\gamma^t~,\quad \gamma=\begin{pmatrix}
        a&b\\c&d
    \end{pmatrix}\in\mr{SL}(2,\IZ)~.
\end{eqsp}
The contour $\CC_{m,n,\ell}$ may be expressed as
\begin{eqsp}
\tau_2'=\frac{2n'}{\varepsilon}, \quad \sigma_2'=\frac{2m'}{\varepsilon}, \quad z_2'=-\frac{\ell'}{\varepsilon}~.    
\end{eqsp}
The vanishing of the contour integral can now be shown as before by replacing $z$ by $z'$ and $\ell$ by $\ell'$.
\def\ngeq{\ge \hskip -.12in / \hskip .1in}

This proves the first part of \refb{edstarT}. To prove the second part we need to
show that $\wt d^*(m,n,\ell)$ vanishes if $T\ngeq \mathbb{0}$ where the symbol
$T\ngeq \mathbb{0}$ means that at least one eigenvalue of $T$ is strictly negative.
We shall prove this by showing 
that the sum 
\be \label{edefFOtilde}
\sum_T \wt d^*(T) e^{2\pi i {\rm Tr}(\Omega T)}\, ,
\ee
would diverge if $\wt d^*(T)$ had been non-zero
for any $T\ngeq 0$.
In particular, we will show that \refb{edefFOtilde} would diverge on the subspace 
\begin{equation}\label{eq:theta_div_space_pre}
    \{\Omega\in\IH_2: \Omega=i\, c\, \mathds{1}\}~,
\end{equation}
for any positive constant $c$.
Let $T^{(0)}$ be a matrix such that $T^{(0)}\not\geq 0$ and $d^*(T^{(0)})\ne 0$.
Then without the loss of generality, we can assume that 
\begin{eqsp}
    T^{(0)}_{11}<0~.
\end{eqsp}
If not, then 
since $T^{(0)}\not\geq\mathbb{0}$, we can pick a primitive integer vector $v\in\IZ^2$ such that $v^tT^{(0)}v<0$. Then there exists another primitive vector $w\in\IZ^2$ such that $\mr{det}[v,w]=1$. Consider $\gamma=[v,w]\in\mathrm{SL}(2,\IZ)$. Then the $(1,1)$-entry of $\gamma^t T^{(0)}\gamma$ is negative:
\begin{eqnarray}
(\gamma^t T^{(0)}\gamma)_{11}=v^tT^{(0)}v<0~.    
\end{eqnarray}
Now, from the manifest $\mr{PSL}(2,\IZ)$-invariance of $S(\Omega)$ and $\Phi_{10}^{-1}$, we see that $\wt d^*(T)$ is PSL(2,$\mathbb{Z}$) invariant and hence
$\widetilde{d}^*(\gamma^t T^{(0)}\gamma)=\widetilde{d}^*(T^{(0)})\neq 0$.
So we can declare $\gamma^t T^{(0)}\gamma$ as the new $T^{(0)}$ and proceed with the analysis.

Let $G^{(0)}\subset \mr{PSL}(2,\mathbb{Z})$ be defined as
\be
\gamma\in G^{(0)} \quad \hbox{iff} \quad \gamma^t T^{(0)} \gamma = T^{(0)}\, .
\ee
Due to the PSL(2,$\mathbb{Z}$)-invariance of $\widetilde{d}^*(T)$, 
the sum over $T$ in \refb{edefFOtilde} must include a sum of the form
\be
\wt d^*(T^{(0)}) \sum_{g\in G^{(0)}\backslash\mr{PSL}(2,\mathbb{Z})} e^{2\pi i \mr{Tr}(\Omega g^t T^{(0)} g)}
= \wt d^*(T^{(0)}) \sum_{g\in G^{(0)}\backslash\mr{PSL}(2,\mathbb{Z})} e^{-2\pi c \mr{Tr}( g^t T^{(0)} g)}\, ,
\ee
where we have used \refb{eq:theta_div_space_pre}.
Now consider restricting the sum over $g$ to a subset $\Gamma^+_\infty$ of $\mr{SL}(2,\mathbb{Z})$: 
\begin{eqsp}
    \Gamma^+_\infty:=\left\{\begin{pmatrix}
        1&k\\0&1
    \end{pmatrix}:k\in\IN\right\}~.
\end{eqsp} 
Then, we have 
\begin{eqsp}\label{eq:theta_div_proof}
   &  \sum_{g\in G^{(0)}\backslash\mr{PSL}(2,\mathbb{Z})} e^{-2\pi c \mr{Tr}( g^t T^{(0)} g)}
   \ge  \sum_{g\in \Gamma^+_\infty} e^{-2\pi c \mr{Tr}( g^t T^{(0)} g)}
    = \, e^{-2\pi (T_{11}+T_{22})}\sum_{k=1}^\infty e^{-4\pi kT_{12}}e^{-2\pi k^2T_{11}}=\infty~,
\end{eqsp}
since $T_{11}<0$. This shows that the convergence of $\wt F(\Omega)$ in the entire 
region $\det \rm Im \, \Omega>1/4$ requires
$\wt{d}^*(T^{(0)})$ to vanish.
Since by theorem \ref{thm:F_rel_Igusa} we know that the expansion of 
$\wt F(\Omega)$ converges in the region
$\det\,$Im$\,\Omega>1/4$, we conclude that $\wt d^*(T)$ must vanish for $T\not\geq \mathbb{0}$.

This shows that
\be
\wt d^*(T) = d^*(T) \quad \forall \ T, 
\ee
and 
\be
\wt F(\Omega) = \sum_T \, d^*(T) \, e^{2\pi i {\rm Tr}(\Omega T)} = F(\Omega)\, .
\ee
Since $\wt F(\Omega)$ is free from singularities in the domain $\det\,$Im$\,\Omega>1/4$
and hence has a convergent expansion, $F(\Omega)$ must also have a convergent expansion in this
domain.
\end{proof}

\noindent{\bf Acknowledgments:}
A.B. and R.K.S. would like to acknowledge the hospitality of ICTS Bangalore where part of this work was done. R.K.S. would like to thank Karam Deo Shankhadhar for some useful discussions about Siegel modular forms. The work of R.K.S. is supported by the US Department of Energy under grant 
DE-SC0010008. The work of A.B. was supported by grant MTR/2019/000582 from the SERB, 
Government of India. The work of A.S. was supported by ICTS-Infosys Madhava Chair 
Professorship and the Department of Atomic Energy, Government of India, under project no. RTI4001.

\bibliography{main.bib}

\end{document}